\documentclass[11pt,a4paper]{article}
\usepackage{hyperref}
\usepackage{amsmath,amsfonts,amssymb,amsthm,mathtools}
\usepackage{authblk}
\usepackage{latexsym,graphicx}
\usepackage{dsfont}
\usepackage{amscd}
\usepackage[dvipsnames]{xcolor}
\usepackage{extarrows} 
\numberwithin{equation}{section} %To label equations per section
\usepackage{yfonts}
\usepackage{authblk}
\usepackage{appendix}

\allowdisplaybreaks

\newtheorem{theorem}{Theorem}[section]

\newtheorem{lemma}{Lemma}[section]
\newtheorem{proposition}{Proposition}[section]
\theoremstyle{remark}
\newtheorem{remark}{Remark}[section]

\def\Xint#1{\mathchoice
{\XXint\displaystyle\textstyle{#1}}%
{\XXint\textstyle\scriptstyle{#1}}%
{\XXint\scriptstyle\scriptscriptstyle{#1}}%
{\XXint\scriptscriptstyle\scriptscriptstyle{#1}}%
\!\int}
\def\XXint#1#2#3{{\setbox0=\hbox{$#1{#2#3}{\int}$}
\vcenter{\hbox{$#2#3$}}\kern-.5\wd0}}

\def\pvint{\Xint-}
\newcommand{\wedgecircold}{\mathrel{
\scalebox{0.8}{\hspace{-3.65pt}\setbox0\hbox{$\wedge$}
							\rlap{\hbox to \wd0{\hss \hspace{3.65pt}\raisebox{5.5pt}{$\circ$} \hss}}\box0} }}

\newcommand{\im}{\mathrm{Im} }

\newcommand{\ee}{\mathrm{e}}
\newcommand{\ii}{\mathrm{i}}
\newcommand{\dd}{\mathrm{d}}

\newcommand{\R}{{\mathbb R}}
\newcommand{\C}{{\mathbb C}}
\newcommand{\Z}{{\mathbb Z}}

\newcommand{\ve}{\mathbf{e}}
\newcommand{\vf}{\mathbf{f}}
\newcommand{\vg}{\mathbf{g}}
\newcommand{\vh}{\mathbf{h}}

\newcommand{\vM}{\mathbf{M}}

\newcommand{\cE}{\mathcal{E}}
\newcommand{\cN}{\mathcal {N}}
\newcommand{\cT}{\mathcal{T}}
\newcommand{\cU}{\mathcal{U}}
\newcommand{\cF}{\mathcal{F}}
\newcommand{\cG}{\mathcal{G}}

\newcommand{\tT}{\tilde{T}}

\newcommand{\mA}{\mathsf{A}}
\newcommand{\mB}{\mathsf{B}}
\newcommand{\mC}{\mathsf{C}}
\newcommand{\mD}{\mathsf{D}}

\newcommand{\mU}{\mathsf{U}}
\newcommand{\mV}{\mathsf{V}}
\newcommand{\mP}{\mathsf{P}}
\newcommand{\mQ}{\mathsf{Q}}

\newcommand{\mM}{\mathsf{M}}
\newcommand{\mF}{\mathsf{F}}
\newcommand{\mG}{\mathsf{G}}
\newcommand{\mX}{\mathsf{X}}

\newcommand{\cA}{\mathcal{A}}

\newcommand{\cV}{\mathcal{V}}
\newcommand{\ocomma}{\,\overset{\circ}{,}\,}

\title{Elliptic soliton solutions of the  \\
spin non-chiral intermediate long-wave equation}
 
\author[1]{Bjorn K. Berntson}
\author[1,2]{Edwin Langmann}
\author[3]{Jonatan Lenells}
\date{\today}

\affil[1]{Department of Physics, KTH Royal Institute of Technology, SE-106 91 Stockholm, Sweden}
\affil[2]{Nordita, KTH Royal Institute of Technology and Stockholm University, SE-106 91 Stockholm, Sweden}
\affil[3]{Department of Mathematics, KTH Royal Institute of Technology, SE-100 44 Stockholm, Sweden}

\begin{document}

\maketitle
 
\begin{abstract}
We construct elliptic multi-soliton solutions of the spin non-chiral intermediate long-wave (sncILW) equation with periodic boundary conditions. 
These solutions are obtained by a spin-pole ansatz including a dynamical background term; we show that this ansatz solves the periodic sncILW equation 
provided the spins and poles satisfy the elliptic $A$-type spin Calogero-Moser (sCM) system with certain constraints on the initial conditions. The key to this result is a B\"acklund transformation for the elliptic sCM system which
 includes a non-trivial dynamical background term. 
We also present solutions of the sncILW equation on the real line and of the spin Benjamin-Ono equation which generalize previously obtained solutions by allowing for a non-trivial background term. 
\end{abstract} 

\section{Introduction}
\label{sec:intro} 
In a recent paper \cite{berntsonlangmannlenells2022}, we introduced and solved new soliton equations related to the $A$-type spin Calogero-Moser (sCM) systems of Gibbons and Hermsen \cite{gibbons1984} (see also \cite{wojciechowski1985}). One of these equations, the spin non-chiral intermediate long wave (sncILW) equation, was shown to have multi-soliton solutions with dynamics described by the hyperbolic sCM system. 
In this paper, we generalize these solutions to the periodic case. More specifically, we construct periodic solutions of the sncILW equation with dynamics described by the elliptic sCM system \cite{krichever1995spin}.
This generalization is non-trivial in several regards; in particular, our solutions include a dynamical background term which, as we show, provides a non-trivial generalization even in the hyperbolic limit when the spatial period becomes infinite.  
 We also present corresponding generalizations of known solutions to the spin Benjamin-Ono (sBO) equation introduced in \cite{berntsonlangmannlenells2022} and, in this way, obtain the full correspondence between sCM models and soliton equations conjectured in \cite{berntsonlangmannlenells2022}. 
 
A prominent feature of the sncILW equation is its nonlocality, which arises through integral operators (see \eqref{eq:sncILW}--\eqref{eq:TTe} below). While soliton equations with this feature have been studied for a long time (the classical examples are the Benjamin-Ono \cite{benjamin1967,ono1975} and intermediate long wave \cite{joseph1977,kubota1978} equations), there has recently been considerable interest in constructing and analyzing novel nonlocal soliton equations; see, for instance, \cite{abanov2009,gerard2010,zhou2015,lenzmann2018,gormley2022,gerard2022,rui2022,takasaki2022}. In this context, in addition to resolving a conjecture posed in \cite{berntsonlangmannlenells2022}, this paper serves to exhibit the sncILW equation with periodic boundary conditions as an interesting system worthy of further study.

Throughout this paper, we denote by $\zeta(z)$ and $\wp(z)$ the usual Weierstrass $\zeta$- and $\wp$-functions of a complex variable $z$ with half-periods $(\ell,\ii\delta)$, with $\ell>0$ and $\delta>0$ fixed parameters and $\ii\coloneqq\sqrt{-1}$ (for the convenience of the reader, we give the definitions of these functions in Appendix~\ref{app:elliptic}). We find it convenient to use the following variants of $\zeta(z)$, 
\begin{equation} 
\label{eq:zetaj}
\begin{split} 
\zeta_1(z)\coloneqq &\; \zeta(z)-\frac{\eta_1}{\ell}z,\quad   \eta_1\coloneqq \zeta(\ell), \\
\zeta_2(z)\coloneqq &\; \zeta(z)-\frac{\eta_2}{\ii\delta}z,\quad   \eta_2\coloneqq \zeta(\ii\delta), 
\end{split} 
\end{equation} 
and the functions 
\begin{equation}\label{eq:wp2f2}
\begin{split}
 \wp_2(z)\coloneqq&  -\zeta_2'(z) = \wp(z)+\frac{\eta_2}{\ii\delta},\\
\varkappa(z)\coloneqq&\; \zeta_2(z)^2-\wp_2(z);
\end{split} 
\end{equation}
see Appendix~\ref{app:elliptic} for details. 
The function $\zeta_1(z)$ is $2\ell$-periodic (but not $2\ii\delta$-periodic) and the function $\zeta_2(z)$ is $2\ii\delta$-periodic (but not $2\ell$-periodic).
Note that $\varkappa(z)$ reduces to a constant in the limits $\ell\to\infty$ and/or $\delta\to\infty$ (as can been seen by evaluating $\varkappa(z)$ in \eqref{eq:wp2f2} using the corresponding degenerations of the functions $\wp_2(z)$ and $\zeta_2(z)$ presented below in \eqref{eq:alpha} and \eqref{eq:V}, respectively). This is one reason why the cases treated in \cite{berntsonlangmannlenells2022} are significantly easier than the elliptic case treated in the present paper.
We also note that $\zeta_2(z)=\zeta_1(z)+\gamma_0 z$ with the constant
\begin{equation}\label{eq:cc}
\gamma_0\coloneqq \frac{\pi}{2\ell\delta}.
\end{equation}
The constant $\gamma_0$ is non-zero only if both $\ell$ and $\delta$ are finite; this is another reason why the elliptic case is more complicated than the cases treated in \cite{berntsonlangmannlenells2022}.

\subsection{Periodic sncILW equation}
For $d$ a fixed positive integer, we denote by $\C^{d\times d}$ the algebra of complex $d\times d$ matrices. 
The periodic sncILW equation describes the time evolution of two $\C^{d\times d}$-valued functions $\mU=\mU(x,t)$ and $\mV=\mV(x,t)$ of $x\in\R$ and $t\in\R$ as follows, 
\begin{equation} 
\label{eq:sncILW} 
\begin{aligned} 
\mU_t &+ \{\mU,\mU_x\} + T\mU_{xx}+\tT \mV_{xx} +\ii [\mU,T\mU_x]+\ii [\mU,\tT \mV_x] =0,\\
\mV_t &- \{\mV,\mV_x\} - T\mV_{xx}-\tT \mU_{xx} +\ii [\mV,T\mV_x]+\ii[\mV,\tT \mU_x] =0,
\end{aligned} 
\end{equation} 
together with the requirement that both functions are $2\ell$-periodic, $\mU(x+2\ell,t)=\mU(x,t)$ and $\mV(x+2\ell,t)=\mV(x,t)$, where $[\cdot,\cdot]$ and $\{\cdot,\cdot\}$ denote the commutator and anti-commutator of square matrices, respectively, and $T$ and $\tT$ are integral operators acting on $2\ell$-periodic functions $f(x)$ of $x\in\R$ as 
\begin{equation}\label{eq:TTe}
\begin{split}
(Tf)(x)\coloneqq &\; \frac{1}{\pi} \pvint_{-\ell}^{\ell} \zeta_1(x'-x) f(x')\,\mathrm{d}x',    \\
(\tilde{T}f)(x) \coloneqq &\; \frac{1}{\pi}\int_{-\ell}^{\ell} \zeta_1(x'-x+\ii\delta)f(x')\,\mathrm{d}x',
\end{split}
\end{equation}
where the dashed integral indicates a principal value prescription and $T$ and $\tilde{T}$ act component-wise on matrix-valued functions. Note that, for $d=1$, the sncILW equation reduces to the periodic non-chiral ILW equation introduced and studied by us in \cite{berntson2020,berntsonlangmann2021}. 

\subsection{Main result}\label{subsec:mainresult}
The solutions of the periodic sncILW equation \eqref{eq:sncILW} that we construct have the form
\begin{equation}\label{eq:ansatz1}
\begin{split}
\mU(x,t)=&\; \ee^{\ii \gamma_0 (\mP +\mP^\dag) t}\mU_0(x,t)\ee^{-\ii \gamma_0 (\mP+\mP^\dag) t},\\
\mV(x,t)=&\; \ee^{\ii \gamma_0 (\mP+\mP^\dag) t}\mV_0(x,t)\ee^{-\ii \gamma_0 (\mP+\mP^\dag) t},
\end{split}
\end{equation}
where 
\begin{equation}\label{eq:ansatz2}
\begin{split}
\mU_0(x,t)=&\;\mM(t)+\ii \sum_{j=1}^N \mP_j(t)   \zeta_2(x-a_j(t)-\ii\delta/2) -\ii \sum_{j=1}^N \mP_j^{\dag}(t)    \zeta_2(x-a_j^*(t)+\ii\delta/2), \\
\mV_0(x,t)=&\; -\mM(t)-\ii \sum_{j=1}^N \mP_j(t)    \zeta_2(x-a_j(t)+\ii\delta/2)+\ii\sum_{j=1}^N \mP_j^{\dag}(t)    \zeta_2(x-a_j^*(t)-\ii\delta/2).
\end{split}
\end{equation}
Here, the time-dependent variables $\mM(t)\in\C^{d\times d}$,  $a_j(t)\in \C$, and $\mP_j(t)\in \C^{d\times d}$ are such that 
\begin{equation}\label{eq:Pdef}
\mP \coloneqq \sum_{j=1}^N \mP_j(t) 
\end{equation}
is time-independent. We refer to $\mM(t)$ as the background and to $a_j(t)$ and $\mP_j(t)$ as pole and spin degrees of freedom, respectively.
Our main result is that, by setting (our notation is explained in Section~\ref{subsec:notation} below)
\begin{equation}
\label{eq:Pj}
\mP_j(t) = |e_j(t)\rangle\langle f_j(t)|
\end{equation}
with $|e_j(t)\rangle$ and $\langle f_j(t)|$ vectors in a $d$-dimensional complex vector space $\cV$ and its dual $\cV^*$, respectively, \eqref{eq:ansatz1}--\eqref{eq:Pdef}  gives an exact solution of the periodic sncILW equation \eqref{eq:sncILW} provided the following conditions are fulfilled: 
\begin{enumerate}
\item[(i)] The dynamical variables $\{a_j,|e_j\rangle,\langle f_j|\}_{j=1}^N =\{a_j(t),|e_j(t)\rangle,\langle f_j(t)|\}_{j=1}^N$ evolve in time according to the following equations, 
\begin{align}
\label{eq:sCMaIV}
\ddot{a}_j=&\;  -4\sum_{k\neq j}^N \langle f_j|e_k\rangle\langle f_k|e_j\rangle \wp_2'(a_j-a_k), \\
\begin{split}\label{eq:sCMefIV}
|\dot{e}_j\rangle = &\;  2\ii\sum_{k\neq j}^N |e_k\rangle\langle f_k|e_j\rangle \wp_2(a_j-a_k)   ,         \\
\langle \dot{f}_j| = &\; -2\ii \sum_{k\neq j}^N \langle f_j|e_k \rangle\langle f_k| \wp_2(a_j-a_k)
\end{split}
\end{align}
for $j=1,\ldots,N$. 

\item[(ii)]  The dynamics of the background $\mM=\mM(t)$ is given by
\begin{equation}
\label{eq:MdotHermitian}
\dot{\mM} =-\frac12\sum_{j=1}^{N} \sum_{k\neq j}^{N} [\mP_j,\mP_k]   \varkappa'(a_j-a_k)+\frac12\sum_{j=1}^{N} \sum_{k\neq j}^{N} [\mP_j^{\dag},\mP_k^{\dag}]   \varkappa'(a_j^*-a_k^*),
\end{equation}
where $\mP_j=\mP_j(t)$.  

\item[(iii)]   At time $t=0$, the following conditions are fulfilled, 
\begin{equation}
\label{eq:sCMconstraint}
\langle e_{j} | f_{j} \rangle=1, 
\end{equation}
\begin{align}
\label{eq:BThermitian}
\dot a_j\langle f_{j}| =&\; 2\langle f_{j}| \mM+2\ii \sum_{k\neq j}^N  \langle f_{j}|\mP_{k} \zeta_2(a_{j}-a_{k}) -2\ii \sum_{k=1}^N  \langle f_{j}| \mP_{k}^{\dag}   \zeta_2(a_{j}-a_{k}^*+\ii\delta), 
\end{align}
\begin{equation}
\label{eq:aj}
-\frac{3\delta}{2}<\im(a_{j})<-\frac{\delta}{2}
\end{equation}
for $j=1,\ldots,N$, together with 
\begin{equation}
\label{eq:mMhermitian}
\mM^{\dag}=\mM
\end{equation}
and
\begin{equation}
\label{eq:mPhermitian}
\mP^{\dag}=\mP.
\end{equation}
\item[(iv)]  The time $t$ is small enough that the poles neither leave the strip defined in \eqref{eq:aj} nor collide (see Theorem~\ref{thm:sncILW} for a more precise formulation; as will be discussed, this is a technical condition needed in our proof but which probably can be ignored). 
\end{enumerate}

Several remarks are in order. 

\begin{enumerate} 
\item One can check that \eqref{eq:Pj}   and \eqref{eq:sCMefIV}  Â indeed   imply that $\mP$  in \eqref{eq:Pdef} is time-independent.
\item It is important to note that \eqref{eq:sCMaIV} and \eqref{eq:sCMefIV} are the time evolution equations of the elliptic sCM model \cite{gibbons1984}. (Note that the elliptic sCM model is usually defined with the standard Weierstrass $\wp$-function $\wp(z)$ instead of $\wp_2(z)$; however, this difference is irrelevant since the system of equations \eqref{eq:sCMaIV}--\eqref{eq:sCMefIV}
is invariant under the transformation 
\begin{equation} 
|e_j(t)\rangle \to \ee^{2\ii c \mP t}|e_j(t)\rangle ,  \quad \langle f_j(t)| \to \langle f_j(t)|\ee^{-2\ii c \mP t},\quad \wp_2(z)\to \wp_2(z)+c
\end{equation} 
with $\mP$ in \eqref{eq:Pdef} Hermitian, for arbitrary $c\in\R$.)

\item We emphasize that the conditions  \eqref{eq:sCMconstraint}--\eqref{eq:mPhermitian} are constraints on initial conditions. If \eqref{eq:sCMconstraint}--\eqref{eq:mPhermitian} are fulfilled at time $t=0$, then the time evolution equations \eqref{eq:sCMaIV}--\eqref{eq:MdotHermitian} guarantee that \eqref{eq:sCMconstraint}--\eqref{eq:mPhermitian} hold true for all times (this is easy to check for \eqref{eq:sCMconstraint}, \eqref{eq:mMhermitian} and \eqref{eq:mPhermitian}, guaranteed by our assumptions for \eqref{eq:aj}, and proved in Proposition~\ref{prop:equiv} in Section~\ref{sec:BT} for \eqref{eq:BThermitian}).

\item While the solutions given above are Hermitian, $\mU=\mU^\dag$ and $\mV=\mV^\dag$, we will actually prove a more general result providing non-Hermitian solutions of the periodic sncILW equation and obtain the Hermitian solutions as a corollary; see Theorem~\ref{thm:sncILW} for the general result. 
We emphasize the Hermitian solutions here since they are easier to state and probably more interesting in physics applications.

\item It not obvious but true that the functions $\mU(x,t)$ and $\mV(x,t)$ in \eqref{eq:ansatz1}--\eqref{eq:Pdef} are $2\ell$-periodic as functions of $x$; to see this, insert $\zeta_2(z)=\zeta_1(z)+\gamma_0 z$  where $\zeta_1(z)$ is $2\ell$-periodic, and observe that the potentially dangerous term $\propto \gamma_0 x$ in \eqref{eq:ansatz2} arising from this insertion vanishes due to the constraint \eqref{eq:mPhermitian}. 
\item We will show in Section~\ref{subsec:multi-solitons} that the constraints  \eqref{eq:sCMconstraint}--\eqref{eq:mMhermitian} can be solved by linear algebra. 

\item For $\ell\to\infty$ and in the special case $\mM=0$, the solutions above reduce to the multi-solition solutions of the sncILW equation on the real line obtained in \cite{berntsonlangmannlenells2022}; it is important to note that one can allow for a non-zero $\mM$ also in the limit $\ell\to\infty$ and, in this way, our solutions here generalize the ones in \cite{berntsonlangmannlenells2022} even in the case $\ell\to\infty$. 

\item The key to our result is a generalization of the B\"acklund transformations of the sCM model in \cite{berntsonlangmannlenells2022} to the elliptic case which is non-trivial since it requires the presence of a non-trivial background $\mM$; with this generalization, we obtain a complete correspondence between sCM models and soliton equations, as anticipated  in \cite{berntsonlangmannlenells2022} (we discuss this point in more detail in Section~\ref{sec:results}). We also mention closely related earlier work on B\"acklund transformations of (s)CM systems \cite{wojciechowski1982,gibbons1983}.

\item The results in this paper add further support to the conjecture in \cite{berntsonlangmannlenells2022} that the periodic sncILW equation is an integrable system. It is important to note that this would be an elliptic integrable system whose analysis indispensably involves elliptic functions, and this presents challenges not present in the limiting cases $\ell\to\infty$ and/or $\delta\to\infty$ treated in \cite{berntsonlangmannlenells2022}. 

\item As discussed, the constraint \eqref{eq:mPhermitian} is required for the solutions \eqref{eq:ansatz1}--\eqref{eq:Pdef} to be $2\ell$-periodic and, in this sense, it can be regarded as a balancing condition. 
We observe that, in the limiting cases  $\ell\to\infty$ and/or $\delta\to\infty$, the constraint \eqref{eq:mPhermitian} can be ignored; otherwise, the result (with non-zero $\mM$) remains true as it stands in these limits (this can be seen by going through our proof of Theorem~\ref{thm:sncILW} and replacing $\varkappa(z)$ by a constant). 
\end{enumerate}

\subsection{Plan of the paper}
\label{subsec:plan} 
In Section~\ref{sec:results}, we formulate and discuss our main result, Theorem~\ref{thm:sncILW}, which provides, in general, non-Hermitian elliptic soliton solutions to the periodic sncILW equation; corresponding results for the sBO equation are also given. 
In Section \ref{section3}, we show that the functions $\mU(x,t)$ and $\mV(x,t)$ in \eqref{eq:ansatz1} solve the periodic sncILW equation \eqref{eq:sncILW} provided that a certain first-order system is satisfied. 
In Section \ref{sec:BT}, we establish a new B\"{a}cklund transformation for the elliptic sCM system.
Using these results, Theorem~\ref{thm:sncILW} is then proved in Section~\ref{section5}. 
In Section~\ref{sec:examples}, we show how to solve the constraints of Theorem~\ref{thm:sncILW} to generate initial data for our $N$-soliton solutions, paying particular attention to the one-soliton case. 
The definitions and functional identities for the Weierstrass elliptic functions that we use are collected in Appendix~\ref{app:special}. 
The proofs of two important lemmas stated in Section \ref{sec:BT} are deferred to Appendix \ref{app:Proofs}.

\subsection{Notation}
\label{subsec:notation}
We follow \cite{gibbons1984} and use the Dirac bra-ket notation \cite{dirac1939} to write our solutions and relate them to the elliptic sCM system. 
In particular, we denote vectors in a $d$-dimensional complex vector space $\cV$ by $|e\rangle$ and vectors in the dual space $\cV^*$ by $\langle f|$. Readers not familiar with this notation can identify $|e\rangle$ with $(e_{\mu})_{\mu=1}^d\in \C^d$, $\langle f|$ with $(f_{\mu}^*)_{\mu=1}^d\in\C^d$ where $*$ is complex conjugation, $\langle f|e\rangle$ with the scalar product $\sum_{\mu=1}^d f_{\mu}^*e_{\mu}$, and $|e\rangle\langle f|$ with the matrix $(e_{\mu}f_{\nu}^*)_{\mu,\nu=1}^d$; in particular, $\cV\otimes\cV^*$ is naturally identified with $\C^{d\times d}$. We denote Hermitian conjugation by $\dag$; note that $|e\rangle^\dag=\langle e|$, $\langle f|^\dag = |f\rangle$ and $(|e\rangle\langle f|)^\dag = |f\rangle\langle e|$ for all $|e\rangle\in\cV$ and $\langle f|\in\cV^*$. 

We use the shorthand notation $\sum_{k\neq j}^N$ for sums $\sum_{k=1,k\neq j}^N$, etc. Dots indicate differentiation with respect to time $t$ and primes indicate differentiation with respect to the argument of the function. 

\section{Results}
\label{sec:results}
As mentioned already in the introduction, the periodic sncILW equation is the elliptic case in a general correspondence between sCM models and soliton equations proposed in \cite{berntsonlangmannlenells2022}. More specifically, there are four cases in this correspondence which, on the sCM side, can be distinguished by the following special functions 
\begin{equation} 
\label{eq:V} 
V(z) \coloneqq \begin{cases} 
1/z^2 & \text{(I: rational case)} \\
(\pi/2\ell)^2/\sin^{2}(\pi z/2\ell)   & \text{(II: trigonometric case)} \\
(\pi/2\delta)^2/\sinh^{2}(\pi z/2\delta)  & \text{(III: hyperbolic case)} \\
\wp_2(z)   & \text{(IV: elliptic case)}
\end{cases}
\end{equation} 
which are the well-known two-body interaction potentials in $A$-type Calogero-Moser systems (see \cite{olshanetsky1981} for review). On the soliton side, the functions 
\begin{equation} 
\label{eq:alpha} 
\alpha(z) \coloneqq \begin{cases} 
1/z & \text{(I: rational case)}\\
(\pi/2\ell)\cot(\pi z/2\ell)  & \text{(II: trigonometric case)}\\
(\pi/2\delta)\coth(\pi z/2\delta)  & \text{(III: hyperbolic case)}\\
\zeta_2(z)  & \text{(IV: elliptic case)}
\end{cases} 
\end{equation} 
are the building blocks in the spin-pole ansatz we use to solve the corresponding soliton equations. Note that $V(z)=-\alpha'(z)$ in all cases. Moreover, the elliptic case (IV) is most general,  and it reduces to the cases I, II and III in the limits $(\ell,\delta)\to(\infty,\infty)$, $\delta\to\infty$ (keeping $\ell$ finite), and $\ell\to\infty$  (keeping $\delta$ finite), respectively. Nevertheless, Cases I--III are interesting in their own right since, first, they are often sufficient in applications, and second, they are significantly simpler and thus allow for more general results that are not directly obtainable as limits of results for Case IV. 

We now give the soliton equations corresponding to Cases I--IV. 
Cases I and II correspond to the sBO equation given by 
\begin{equation} 
\label{eq:sBO} 
\mU_t + \{\mU,\mU_x\} +  H\mU_{xx}  +\ii [\mU,H\mU_x]=0 
\end{equation} 
with a $\C^{d\times d}$-valued function $\mU=\mU(x,t)$ of $x\in\R$ and $t\in\R$, with $H$ the Hilbert transform; Case I corresponds to the sBO equation on the real line where $\mU(x,t)$ has suitable decaying conditions at $x\to\pm\infty$ and 
\begin{equation}
\label{eq:H}  
(Hf)(x)\coloneqq \frac1{\pi}\pvint_{\R} \frac1{x'-x} f(x')\,\dd{x'} , 
\end{equation} 
and Case II corresponds to the periodic sBO equation, where $\mU(x+2\ell,t)=\mU(x,t)$, and with the periodic Hilbert transform 
\begin{equation}
\label{eq:Hp}  
(Hf)(x)\coloneqq \frac1{2\ell}\pvint_{-\ell}^{\ell} \cot\left(\frac{\pi}{2\ell}(x'-x) \right)f(x')\,\dd{x'}  . 
\end{equation} 
Case III corresponds to sncILW equation \eqref{eq:sncILW} on the real line, with functions $\mU(x,t)$ and $\mV(x,t)$ of $x\in\R$ and $t\in\R$ satisfying suitable decaying conditions at $x\to\pm\infty$,  and with the integral operators $T$ and $\tilde{T}$ defined as 
\begin{equation} 
\label{eq:TT}
\begin{split} 
(Tf)(x) &= \frac1{2\delta}\pvint_{\R}\coth\left(\frac{\pi}{2\delta}(x'-x)\right)f(x')\,\dd{x'}, \\
(\tT f)(x) &= \frac1{2\delta}\int_{\R}\tanh\left(\frac{\pi}{2\delta}(x'-x)\right)f(x')\,\dd{x'}. 
\end{split} 
\end{equation} 
Finally, Case IV, which is the most general, corresponds to the periodic sncILW equation \eqref{eq:sncILW} with $T$ and $\tilde{T}$ given by \eqref{eq:TTe}. 
Note that, since $\tT\to 0$ and $T\to H$ in the limit $\delta\to\infty$ \cite{berntsonlangmann2020}, the first equation in \eqref{eq:sncILW} reduces to the sBO equation \eqref{eq:sBO} in this limit. Moreover, $T$ and $\tT$ in \eqref{eq:TTe} reduce to $T$ and $\tT$ in \eqref{eq:TT} in the limit $\ell\to\infty$, and $H$ in \eqref{eq:Hp} reduces to $H$ in \eqref{eq:H} in this limit. 
For future reference, we summarize the discussion in the present paragraph as follows, 
\begin{equation} 
\label{eq:soliton}
\begin{cases} 
\text{sBO equation on the real line} & \text{(I: rational case)}\\
\text{periodic sBO equation} & \text{(II: trigonometric case)}\\
\text{sncILW equation on the real line} & \text{(III: hyperbolic case)}\\
\text{periodic sncILW equation} & \text{(IV: elliptic case)}.
\end{cases} 
\end{equation} 

In this section, we present our solutions of the soliton equations in \eqref{eq:soliton} in all Cases I--IV. The solutions for the Cases I--III that we present are generalizations of solutions obtained already in \cite{berntsonlangmannlenells2022}; the simplest way to prove these generalizations is to adapt the proofs in \cite{berntsonlangmannlenells2022}, as recently done in a thesis by Anton Ottosson \cite{ottosson2022}. This is due to additional constraints appearing in Case~IV which prevent a direct derivation of {\em all}  solutions in Cases~I--III as limits of the general solution in Case~IV. For this reason, and for clarity, the detailed proofs we present in this paper are restricted to Case IV.

\subsection{Solutions of the sncILW equation} 
\label{subsec:IV}
Throughout this subsection, we consider Cases III and IV in \eqref{eq:V}--\eqref{eq:alpha} and \eqref{eq:soliton}; in particular, $V(z)$ and $\alpha(z)$ are as in \eqref{eq:V}  and \eqref{eq:alpha} for the hyperbolic and elliptic cases. Note that  $\varkappa(z)=(\pi/2\delta)^2$ (a constant!) in Case III, while $\varkappa(z)$ is the non-trivial function in \eqref{eq:wp2f2} in Case IV. 

Our general solutions of the sncILW equation (including non-Hermitian ones) are defined in terms of two sets of variables satisfying the time evolution equations of the sCM model: 
\begin{align}
\label{eq:sCMa}
\ddot{a}_j=&\; -4\sum_{k\neq j}^N \langle f_j|e_k\rangle\langle f_k|e_j\rangle V'(a_j-a_k), \\
\label{eq:sCMef}
\begin{split}
|\dot{e}_j\rangle = &\;  2\ii\sum_{k\neq j}^N |e_k\rangle\langle f_k|e_j\rangle V(a_j-a_k)   ,         \\
\langle \dot{f}_j| = &\; -2\ii \sum_{k\neq j}^N \langle f_j|e_k \rangle\langle f_k| V(a_j-a_k)
\end{split}
\end{align}
for $j=1,\ldots,N$ and 
\begin{align}
\label{eq:sCMb}
\ddot{b}_j=&\; -4\sum_{k\neq j}^M \langle h_j|g_k\rangle\langle h_k|g_j\rangle V'(b_j-b_k),  \\
\label{eq:sCMgh}
\begin{split}
|\dot{g}_j\rangle = &\;  2\ii\sum_{k\neq j}^M |g_k\rangle\langle h_k|g_j\rangle V(b_j-b_k)   ,         \\
\langle \dot{h}_j| = &\; -2\ii \sum_{k\neq j}^M \langle h_j|g_k \rangle\langle h_k| V(b_j-b_k)
\end{split}
\end{align}
for $j=1,\ldots,M$ (we use the notation in Section~\ref{subsec:notation}). More specifically, the spin-pole ansatz providing solutions of the sncILW equation \eqref{eq:sncILW}, both in the real-line case (III) and the periodic case (IV), is given in terms of these dynamical variables as follows,  
\begin{equation}\label{eq:ansatzgen}
\begin{split}
\mU(x,t)=&\; \ee^{\ii\gamma_0 (\mP+\mQ) t}\mU_0(x,t)\ee^{-\ii \gamma_0(\mP+\mQ)t},\\ \mV(x,t)=&\; \ee^{\ii\gamma_0 (\mP+\mQ)t}\mV_0(x,t)\ee^{-\ii \gamma_0(\mP+\mQ)t},
\end{split}
\end{equation}
where
\begin{equation}
\label{eq:ansatzgen2}
\begin{split}
\mU_0(x,t)=&\;\phantom- \mM(t)+\ii \sum_{j=1}^N  \mP_j(t)  \alpha(x-a_j(t)-\ii\delta/2) -\ii \sum_{j=1}^{M} \mQ_j(t) \alpha(x-b_j(t)+\ii\delta/2), \\
\mV_0(x,t)=&\; -\mM(t)-\ii \sum_{j=1}^N \mP_j(t)\alpha(x-a_j(t)+\ii\delta/2)+\ii\sum_{j=1}^{M} \mQ_j(t)  \alpha(x-b_j(t)-\ii\delta/2),
\end{split}
\end{equation}
with 
\begin{equation} 
\label{eq:PjQj} 
\mP_j(t) = |e_j(t)\rangle\langle f_j(t)|\quad  (j=1,\ldots,N), \quad \mQ_j(t) = |g_j(t)\rangle\langle h_j(t)| \quad (j=1,\ldots,M)
\end{equation} 
such that 
\begin{equation}\label{eq:PQdef}
\mP=\sum_{j=1}^N \mP_j(t),\qquad \mQ\coloneqq 	\sum_{j=1}^{M} \mQ_j(t)
\end{equation}
both are time-independent, together with an additional variable $\mM(t)\in\C^{d\times d}$ describing a non-trivial background. In Case III, $N$ and $M$ are arbitrary and the variable $\mM$ must be constant, whereas in Case IV, we must have $N=M$ and $\mM$ is necessarily dynamical. The precise statement is as follows. 

\begin{theorem}[Non-Hermitian solutions of the sncILW equation] 
\label{thm:sncILW}
For fixed $N,M\in\Z_{\geq 0}$, let $a_j(t)\in\C$, $|e_j(t)\rangle\in\cV$, $\langle f_j(t)|\in\cV^*$ for $j=1,\ldots,N$, 
$b_j(t)\in\C$, $|g_j(t)\rangle\in\cV$, $\langle h_j(t)|\in\cV^*$ for $j=1,\ldots,M$, and $\mM(t)\in\C^{d\times d}$ be functions of $t\in\R$ satisfying the following conditions: (i) Both sets of variables $\{a_j,|e_j\rangle,\langle f_j|\}_{j=1}^N$ and $\{b_j,|g_j\rangle,\langle h_j|\}_{j=1}^M$  satisfy the time evolution equations of the sCM model:  \eqref{eq:sCMa}--\eqref{eq:sCMef} hold true for $j=1,\ldots,N$, and \eqref{eq:sCMb}--\eqref{eq:sCMgh} hold true for $j=1,\ldots,M$. 
(ii) The background $\mM=\mM(t)$ satisfies  
\begin{equation}
\label{eq:Mdotgen}
\dot{\mM} =-\frac12\sum_{j=1}^{N} \sum_{k\neq j}^{N} [\mP_j,\mP_k]   \varkappa'(a_j-a_k)+\frac12\sum_{j=1}^{M} \sum_{k\neq j}^{M} [\mQ_j,\mQ_k]   \varkappa'(b_j-b_k)
\end{equation}
with $\mP_j=\mP_j(t)$ and $\mQ_j=\mQ_j(t)$ given by \eqref{eq:PjQj}. 
(iii) At time $t=0$, the following conditions are fulfilled: first, 
\begin{equation}
\label{eq:sCMconstraint1}
\langle e_{j} | f_{j} \rangle=1\quad (j=1,\ldots,N), \quad 
\langle g_{j} | h_{j} \rangle=1\quad (j=1,\ldots,M), 
\end{equation}
second,  
\begin{equation}
\label{eq:BT}
\begin{split}
&\dot{a}_j \langle f_j|  = 2\langle f_j| \mM   +2\ii \sum_{k\neq j}^N \langle f_j|e_k\rangle\langle f_k|  \alpha(a_j-a_k) -2\ii \sum_{k=1}^M \langle f_j | g_k\rangle \langle h_k|  \alpha(a_j-b_k+\ii\delta) \quad (j=1,\ldots,N),     \\
&\dot{b}_j |g_j\rangle =  2\mM |g_j\rangle -2\ii \sum_{k\neq j}^M |g_k\rangle\langle h_k|g_j\rangle  \alpha(b_j-b_k) + 2\ii \sum_{k=1}^N |e_k\rangle\langle f_k | g_j\rangle \alpha(b_j-a_k+\ii\delta)\quad (j=1,\ldots,M), 
\end{split}
\end{equation}
and third,
\begin{equation}
\label{eq:ajbj}
-\frac{3\delta}{2}<\im(a_{j})<-\frac{\delta}{2} \quad (j=1,\ldots,N),\quad \frac{\delta}{2}<\im(b_{j})<\frac{3\delta}{2}\quad (j=1,\ldots,M) .  
\end{equation}
Then, in Case III, $\mP$ and $\mQ$ in \eqref{eq:PQdef} are both time-independent, and the ansatz \eqref{eq:ansatzgen}--\eqref{eq:PQdef} gives a solution of the sncILW equation on the real line for all times in the interval $t\in[0,\tau)$ if  $\tau>0$ is such that \eqref{eq:ajbj} and
\begin{equation} 
\label{eq:ajak}
a_j \neq a_k \quad (1\leq j<k\leq N),\quad b_j\neq b_k \quad (1\leq j<k\leq M)
\end{equation} 
hold true for all times $t\in[0,\tau)$; in Case IV, the same holds true provided that $N=M$, the conditions in \eqref{eq:ajak} hold $\bmod \,2\ell$, and the following holds true at time $t=0$, 
\begin{equation} 
\label{eq:constraintPjQj}
\mP=\mQ. 
\end{equation} 
\end{theorem}

We give various remarks related to this result. 
\begin{enumerate} 
\item 
\label{rem:technical}
Theorem~\ref{thm:sncILW} gives a Hermitian solution, $\mU=\mU^\dag$ and $\mV=\mV^\dag$, if and only if $N=M$ and the initial conditions satisfy the following further constraints,
\begin{equation} 
\label{eq:hermitian} 
\mM=\mM^{\dag},\quad b_j=a_j^*,\quad |e_j\rangle=\langle h_j|^\dag,\quad \langle f_j|= |g_j\rangle^\dag\quad (j=1,\ldots,N),
\end{equation} 
which imply $\mQ_j=\mP_j^\dag$. 
It is easy to check that the reduction \eqref{eq:hermitian} is consistent; moreover, if we impose it at time $t=0$, it is fulfilled at all times. Thus, in Case IV, we obtain the Hermitian solution of the scnILW equation presented in Section~\ref{sec:intro}.

\item It is important to note the following differences between Cases III and IV: 
First, in Case~III, $\gamma_0=0$, and therefore $\mU(x,t)=\mU_0(x,t)$ and $\mV(x,t)=\mV_0(x,t)$. However, in Case IV, the functions $\mU_0(x,t)$ and $\mV_0(x,t)$ given by the spin-pole ansatz \eqref{eq:ansatzgen2} are related to the solutions $\mU(x,t)$ and $\mV(x,t)$ by a time-dependent similarity transformation determined by the total spin $\mP$.
Second,  in Case IV, the background $\mM(t)$ has non-trivial dynamics, while in Case III, \eqref{eq:Mdotgen} simplifies to  $\dot\mM=0$ (since $\varkappa(z)$ reduces to a constant), i.e., the background is constant in time: $\mM(t)=\mM(0)\eqqcolon\mM_0$. Third, in Case IV, we impose the constraint \eqref{eq:constraintPjQj} on the initial conditions, but this constraint is absent in Case III; as discussed in Section~\ref{subsec:mainresult}, this additional constraint in Case IV is required for the $2\ell$-periodicity of $\mU$ and $\mV$ (the argument given there straightforwardly generalizes to the more general case here).

\item In Case III, our solutions \eqref{eq:ansatzgen} obey the boundary conditions 
\begin{equation} 
\label{eq:limits} 
\begin{split} 
\lim_{x\to\pm \infty}\mU(x,t)= -\lim_{x\to\pm \infty}\mV(x,t)=
\ee^{\ii\gamma_0 (\mP+\mQ) t}\bigg(\mM_0 \pm  \frac{\ii\pi}{2\delta}(\mP-\mQ)\bigg)\ee^{-\ii \gamma_0(\mP+\mQ)t}
\end{split} 
\end{equation} 
(this follows from $\lim_{x\to\pm\infty}\alpha(x-a)\to\pm \pi/2\delta$ for all $a\in\C$); thus, the condition \eqref{eq:constraintPjQj} is equivalent to $\mU$ and $-\mV$ being equal to $\ee^{\ii\gamma_0 (\mP+\mQ) t}\mM_0\ee^{-\ii \gamma_0(\mP+\mQ)t}$ at $x\to \pm\infty$; however, there is no need to impose these conditions in Case III. Thus, our solutions suggest that the sncILW equation on the real line is well-defined for the following boundary conditions: 
\begin{equation} 
 \lim_{x\to\pm \infty}\mU(x,t)= -  \lim_{x\to\pm \infty}\mV(x,t) = \ee^{\ii \gamma_0(\mP+\mQ)t}\mM_{\pm \infty}\ee^{-\ii \gamma_0(\mP+\mQ)t},
\end{equation} 
where $\mM_{\pm \infty}$Â  are time-independent and, in general, such that $\mM_{+\infty}\neq  \mM_{-\infty}$.

\item The solutions of the sncILW on the real line (Case III) obtained in \cite{berntsonlangmannlenells2022} correspond to the special case $\mM=0$; note that this specialization is only possible in Case III (in Case IV, it is prevented by the non-trivial dynamics \eqref{eq:Mdotgen}). 

\item It is interesting to note that \eqref{eq:BT}, together with \eqref{eq:sCMef}, \eqref{eq:sCMgh}, \eqref{eq:Mdotgen} and \eqref{eq:constraintPjQj}, is a B\"acklund transformation of the sCM system in the 
 elliptic case (IV); see Section~\ref{sec:BT}  for precise statements. Moreover, in the limits  $(\ell,\delta)\to(\infty,\infty)$, $\delta\to\infty$, and $\ell\to\infty$, this B\"acklund transformation reduces to B\"acklund transformations for the sCM system in Cases I, II and III, respectively; it is important to note that, in Cases I--III, the constraint  \eqref{eq:constraintPjQj} can be omitted; this B\"{a}cklund transformation is a generalization of the B\"acklund transformation of sCM model in Cases I--III obtained in \cite{berntsonlangmannlenells2022}; the latter correspond to the special case $\mM=0$ where, again, this specialization is possible in Cases I--III, but not in Case IV. 
We also remark that, in the elliptic case (IV), a similar B\"{a}cklund transformation for a certain singular limit of the elliptic sCM system, where $d=2$ and $\langle f_j|e_j\rangle=\langle h_j|g_j\rangle =0$, was recently found by one of the authors in collaboration with Klabbers \cite{berntsonklabbers2022}.

\item
The B\"{a}cklund transformation, \eqref{eq:BT} with \eqref{eq:sCMef}, \eqref{eq:sCMgh}, \eqref{eq:Mdotgen} and \eqref{eq:constraintPjQj}, forms an overdetermined system of ordinary equations (ODEs) whose consistency must be established. This was done for the known B\"{a}cklund transformation for the sCM system \cite{gibbons1984} in Cases I--III in \cite{berntson2022consistency} by constructing functions that measure the departure of the B\"{a}cklund transformation from consistency and showing that they obey a system of linear homogeneous ODEs; if the initial data is consistent with the B\"{a}cklund transformation at $t=0$, consistency will be preserved at future times, under mild assumptions. While the approach of \cite{berntson2022consistency} can be straightforwardly generalized to the B\"{a}cklund transformation in Case IV, leading to essentially the same equations (see \cite[Eq.~(2.43)]{berntson2022consistency}), in this paper we take a more streamlined approach, which allows us to show that the relevant quantities obey a system of linear homogenous ordinary differential equations without deriving its precise form.

\item Theorem \ref{thm:sncILW} is stated under the assumption that the poles remain in the strip defined in \eqref{eq:ajbj} and that no pole collisions occur (see \eqref{eq:ajak}). 
As already mentioned, we believe that these assumptions are unnecessary, and we expect that our soliton solutions can be extended to all times $t\in\R$. To support this expectation, we mention the following known results. First, in the sBO case, a Lax pair is known which can be used to prove that the former condition is satisfied for all times \cite{gerard2022};  we hope that it is possible to generalize this argument to the sncILW case. Second, for the scalar Benjamin-Ono equation, it is known that pole collisions occur, but they are no problem for the soliton solutions \cite{gerard2021}. A third reason is recent work by G\'{e}rard and Lenzmann on  multi-soliton solutions to a nonlocal nonlinear Schr\"{o}dinger equation \cite{gerardlenzmann2022} (see also \cite{matsuno2002}) governed by a complexification of the rational CM system; in this work, an explicit example of a two-soliton solution is given where (i) the poles collide and (ii) the solution remains valid during and after the collision. Clearly, it would be interesting to prove the more general result.
\end{enumerate} 

\subsection{Solutions of the sBO equation}

In this section, we consider Cases I and II in \eqref{eq:V}--\eqref{eq:alpha} and \eqref{eq:soliton}; in particular, $V(z)$ and $\alpha(z)$ are as in \eqref{eq:V}  and \eqref{eq:alpha} for the rational and trigonometric cases. 

The spin-pole ansatz for solutions of the sBO equation \eqref{eq:sBO}, both in the real-line and periodic cases, is given by
\begin{equation}
\label{eq:ansatzgensBO}
\mU(x,t)=\;\mM_0+\ii \sum_{j=1}^N  \mP_j(t)  \alpha(x-a_j(t)) -\ii \sum_{j=1}^M \mQ_j(t) \alpha(x-b_j(t)) 
\end{equation}
with $\mP_j(t)$ and $\mQ_j(t)$ as in \eqref{eq:PjQj}; as in the hyperbolic case, we can consistently assume that the background $\mM(t)=\mM_0$ is time-independent. 

\begin{theorem}[Non-Hermitian solutions of the sBO equation] 
\label{thm:sBO}
For fixed $N,M\in\Z_{\geq 0}$ and $\mM_0\in\C^{d\times d}$, let $a_j(t)\in\C$, $|e_j(t)\rangle\in\cV$, $\langle f_j(t)|\in\cV^*$ for $j=1,\ldots,N$ and 
$b_j(t)\in\C$, $|g_j(t)\rangle\in\cV$, $\langle h_j(t)|\in\cV^*$ for $j=1,\ldots,M$, satisfy the following conditions: (i) Both sets of variables $\{a_j,|e_j\rangle,\langle f_j|\}_{j=1}^N=\{a_j(t),|e_j(t)\rangle,\langle f_j(t)|\}_{j=1}^N$ and $\{b_j,|g_j\rangle,\langle h_j|\}_{j=1}^M=\{b_j(t),|g_j(t)\rangle,\langle h_j(t)|\}_{j=1}^M$  satisfy the time evolution equations of the sCM model:  
\eqref{eq:sCMa}--\eqref{eq:sCMef} hold true for $j=1,\ldots,N$ and \eqref{eq:sCMb}--\eqref{eq:sCMgh} hold true for $j=1,\ldots,M$. (ii) At time $t=0$, the constraints \eqref{eq:sCMconstraint1}, \eqref{eq:BT} with $\mM(0)=\mM_0$ and $\delta\to 0$, and
\begin{equation}
\im(a_j)<0 \quad (j=1,\ldots,N), \qquad \im(b_j)>0 \quad (j=1,\ldots,M) 	
\end{equation}
hold. 
Then, in Case I, the ansatz \eqref{eq:ansatzgensBO} with \eqref{eq:PjQj} gives a solution of the sBO equation for all times in the interval $t\in[0,\tau)$ if $\tau>0$ is such that \eqref{eq:ajak} holds for all times $t\in[0,\tau)$. In Case II, the same holds true provided the equalities in \eqref{eq:ajak} hold $\bmod \,2\ell$. 
\end{theorem} 

\section{From the sncILW equation to a first-order system}\label{section3}
In this section, we establish conditions under which the ansatz \eqref{eq:ansatzgen} solves the periodic sncILW equation. 
Throughout this section, $V(z)=\wp_2(z)$, $\alpha(z)=\zeta_2(z)$, and $M=N$.

\begin{proposition}\label{prop:solitonsgen}
The functions $\mU(x,t)$ and $\mV(x,t)$ in \eqref{eq:ansatzgen} satisfy the periodic sncILW equation \eqref{eq:sncILW} provided that the equations \eqref{eq:Mdotgen},
\begin{equation}\label{eq:ajdotsolgen}
\begin{split}
\dot{a}_j \mP_j=&\; 2\mP_j\mM +2\ii\sum_{k\neq j}^N \mP_j\mP_k \alpha(a_j-a_k)-2\ii\sum_{k=1}^N \mP_j\mQ_k\alpha(a_j-b_k+\ii\delta), \\
\dot{b}_j \mQ_j=&\; 2\mM \mQ_j -2\ii\sum_{k\neq j}^N \mQ_k\mQ_j \alpha(b_j-b_k)+2\ii\sum_{k=1}^N\mP_k\mQ_j\alpha(b_j-a_k+\ii\delta) 
\end{split}\quad (j=1,\ldots,N),
\end{equation}
and
\begin{equation}\label{eq:Pjdotsolgen}
\begin{split}
\dot{\mP}_j=&\; -2\ii \sum_{k\neq j}^N[\mP_j,\mP_k]V(a_j-a_k), \\
\dot{\mQ}_j=&\; -2\ii \sum_{k\neq j}^N [\mQ_j,\mQ_k]V(b_j-b_k)  
\end{split} \quad (j=1,\ldots,N),
\end{equation}
are satisfied and the constraints \eqref{eq:ajbj}--\eqref{eq:constraintPjQj} and
\begin{equation}\label{eq:Pj2solgen}
\mP_j^2=\mP_j, \quad \mQ_j^2=\mQ_j \quad (j=1,\ldots,N)
\end{equation}
are fulfilled.
\end{proposition}

\begin{proof}

The proof is facilitated by introducing the notation
\begin{equation}\label{eq:circ}
\left(\begin{array}{c} F_1 \\ F_2 \end{array}\right)\circ\left(\begin{array}{c} G_1 \\ G_2 \end{array}\right)\coloneqq \left(\begin{array}{c}
F_1G_1 \\ -F_2G_2
\end{array}
\right)
\end{equation}
for $\C$-valued functions $F_j,G_j$ ($j=1,2$) and the operator 
\begin{equation}\label{eq:cT}
\cT\coloneqq \left(\begin{array}{cc}
T & \tilde{T} \\
-\tilde{T} & -T
\end{array}\right),
\end{equation}
interpreted as a linear operator on vector-valued functions, see \cite{berntson2020}. In the present paper, we use the product $\circ$ defined in \eqref{eq:circ} also for vectors $\cF,\cG$ whose components $\mF_j,\mG_j$ are in $\C^{d\times d}$, and we let
\begin{equation}\label{eq:commutators}
 [\cF\ocomma \cG]\coloneqq \cF\circ\cG- \cG\circ \cF ,\qquad \{\cF\ocomma \cG\}\coloneqq \cF\circ\cG+\cG\circ \cF
\end{equation}
be the corresponding generalizations of the commutator and anti-commutator, respectively. With this notation, the periodic sncILW equation \eqref{eq:sncILW} can be written as
\begin{equation}\label{eq:sncILW2}
\cU_t+\{\cU\ocomma \cU_x\}+\cT\cU_{xx}+\ii[\cU\ocomma \cT\cU_x]=0. 
\end{equation}
Using the shorthand notation
\begin{equation}\label{eq:shorthandsolitons}
(a_j, |e_j\rangle,\langle f_j|,\mP_j ,r_j)\coloneqq \begin{cases}
(a_j, |e_j\rangle,\langle f_j| ,\mP_j ,+1)& j=1,\ldots,N, \\
(b_{j-N}, |g_{j-N}\rangle,\langle h_{j-N}| ,\mQ_{j-N} ,-1)& j=N+1,\ldots,\cN,
\end{cases} \qquad \cN\coloneqq 2N,
\end{equation}
we write the ansatz \eqref{eq:ansatzgen} as
\begin{equation}\label{eq:ansatzSH}
\cU(x,t)= \ee^{\ii\gamma_0(\mP+\mQ)t}\cU_0(x,t)	\ee^{-\ii\gamma_0(\mP+\mQ)t},
\end{equation}
with
\begin{equation}\label{eq:ansatzSH0}
 \cU_0(x,t)= \mM(t)\cE+ \ii\sum_{j=1}^{\cN} r_j \mP_j(t)\cA_{r_j}(x-a_j(t)),
\end{equation}
where
\begin{equation}\label{eq:EA}
\cE \coloneqq \left(\begin{array}{c}
1 \\
-1
\end{array}\right),\qquad 
\cA_{\pm}(z)\coloneqq\left( \begin{array}{c}
+\alpha(z\mp\ii\delta/2) \\
-\alpha(z\pm\ii\delta/2)
\end{array}\right).
\end{equation}

It is important to note that, provided \eqref{eq:Pjdotsolgen} holds, the quantities $\mP$ and $\mQ$ defined in \eqref{eq:PQdef} and appearing in \eqref{eq:ansatzSH} are conserved quantities, 
\begin{subequations}\label{eq:Pj-Qjdot}
\begin{equation}
\dot{\mP} =  \sum_{j=1}^N \dot{\mP}_j
= -2\ii \sum_{j=1}^N\sum_{k\neq j}^{N} [\mP_j,\mP_k]V({a}_j-{a}_k)=0
\end{equation}
and
\begin{equation}
\dot{\mQ}=\sum_{j=1}^N \dot{\mQ}_j=-2\ii  \sum_{j=1}^N\sum_{k\neq j}^{N} [\mQ_j,\mQ_k]V({b}_j-{b}_k)=0,
\end{equation}
\end{subequations}
using the anti-symmetry of the commutator and the fact that $V(z)$ is an even function \eqref{eq:parity}. Thus, assuming \eqref{eq:Pjdotsolgen}, it follows from \eqref{eq:sncILW2} and \eqref{eq:ansatzSH}--\eqref{eq:ansatzSH0} that $\cU$ satisfies \eqref{eq:sncILW2} if and only if $\cU_0$ satisfies
\begin{equation}\label{eq:sncILW3}
\cU_{0,t}+\ii\gamma_0[(\mP+\mQ)\cE\ocomma\cU_0]+\{\cU_0\ocomma \cU_{0,x}\}+\cT\cU_{0,xx}+\ii[\cU_0\ocomma \cT\cU_{0,x}]=0.
\end{equation}

We compute each term in \eqref{eq:sncILW3} with $\cU_0$ given by \eqref{eq:ansatzSH0}. We start with
\begin{equation}\label{eq:Ut}
\cU_{0,t}=\dot{\mM}\cE+\ii\sum_{j=1}^{\cN} r_j \big(\dot{\mP}_j \cA_{r_j}(x-a_j)-\mP_j \dot{a}_j\cA_{r_j}'(x-a_j)\big).
\end{equation}
Next, we compute
\begin{align}\label{eq:UUx1}
\{\cU_0\ocomma\cU_{0,x}\}=&\; \ii \sum_{j=1}^{\cN}r_j  \{\mM,\mP_j\} \cE\circ \cA_{r_j}'(x-a_j)    -\sum_{j=1}^{\cN}\sum_{k=1}^{\cN} r_jr_k \{\mP_j,\mP_k\}\cA_{r_j}(x-a_j)\circ\cA_{r_k}'(x-a_k) \nonumber \\
=&\;  \ii \sum_{j=1}^{\cN}r_j  \{\mM,\mP_j\} \cA_{r_j}'(x-a_j)-2\sum_{j=1}^{\cN} \mP_j^2 \cA_{r_j}(x-a_j)\circ\cA_{r_j}'(x-a_j) \nonumber \\
&\; -\sum_{j=1}^{\cN}\sum_{k \neq j}^{\cN} r_jr_k \{\mP_j,\mP_k\} \cA_{r_j}(x-a_j)\circ\cA_{r_k}'(x-a_k).
\end{align}
To proceed, we need the identities
\begin{equation}\label{eq:Id3}
2\cA_{r_j}(x-a_j)\circ\cA_{r_j}'(x-a_j)=-\cA_{r_j}''(x-a_j)+\cF_{r_j}'(x-a_j)
\end{equation}
and
\begin{align}\label{eq:Id4}
\cA_{r_j}(x-a_j)\circ\cA_{r_k}'(x-a_k)=&\;  -\alpha(a_j-a_k+\ii(r_j-r_k)\delta/2)\cA_{r_k}'(x-a_k)\nonumber \\
&\; -V(a_j-a_k+\ii(r_j-r_k)\delta/2)\big(\cA_{r_j}(x-a_j)-\cA_{r_k}(x-a_k)\big) \nonumber \\
&\;+\frac12\cF_{r_k}'(x-a_k) +\frac12 \varkappa'(a_j-a_k+\ii(r_j-r_k)\delta/2)\cE,
\end{align}
where
\begin{equation}\label{eq:F}
\cF_{\pm}(z)\coloneqq \left(\begin{array}{c} \varkappa(z\mp\ii\delta/2) \\ -\varkappa(z\pm\ii\delta/2)   \end{array}\right).
\end{equation}
The first identity \eqref{eq:Id3} can be obtained by differentiating \eqref{eq:Id1} with respect to $z$ and setting $z=x-a_j \pm r_j\ii\delta/2$ while the second identity \eqref{eq:Id4} can be obtained by differentiating \eqref{eq:Id2} with respect to $c$, setting $a = x$, $b=a_j \pm r_j\ii\delta/2$, and $c=a_k \pm r_k\ii\delta/2$, and using the periodicity property $\alpha(z+2\ii\delta)=\alpha(z)$ and the fact that $\alpha(z)$ is an odd function. Inserting \eqref{eq:Id3} and \eqref{eq:Id4} into \eqref{eq:UUx1} gives
\begin{align}
&\{\cU_0\ocomma\cU_{0,x}\}=  \ii \sum_{j=1}^{\cN}r_j  \{\mM,\mP_j\} \cA_{r_j}'(x-a_j)+\sum_{j=1}^{\cN} \mP_j^2 \cA_{r_j}''(x-a_j)   -\sum_{j=1}^{\cN} \mP_j^2 \cF_{r_j}'(x-a_j) \nonumber \\
& +\sum_{j=1}^{\cN}\sum_{k \neq j}^{\cN} r_jr_k \{\mP_j,\mP_k\}\alpha(a_j-a_k+\ii(r_j-r_k)\delta/2)\cA_{r_k}'(x-a_k) \nonumber  \\
& +\sum_{j=1}^{\cN}\sum_{k \neq j}^{\cN} r_jr_k \{\mP_j,\mP_k\} V(a_j-a_k+\ii(r_j-r_k)\delta/2)\big(\cA_{r_j}(x-a_j)-\cA_{r_k}(x-a_k)\big) \nonumber \\
& -\frac12 \sum_{j=1}^{\cN}\sum_{k\neq j}^{\cN} r_j r_k\{\mP_j,\mP_k\}\cF_{r_k}'(x-a_k)- \frac12\sum_{j=1}^{\cN}\sum_{k\neq j}^{\cN}r_j r_k \{\mP_j,\mP_k\} \varkappa'(a_j-a_k+\ii(r_j-r_k)\delta/2)\cE.
\end{align}
The double sum in the third line and the second double sum in the fourth line vanish by symmetry. Hence, after relabelling summation indices $j\leftrightarrow k$ in the double sum in the second line and rearranging, we are left with
\begin{align}\label{eq:UUx2}
\{\cU_0\ocomma\cU_{0,x}\}=&\;  \ii \sum_{j=1}^{\cN}r_j  \{\mM,\mP_j\} \cA_{r_j}'(x-a_j)+\sum_{j=1}^{\cN} \mP_j^2 \cA_{r_j}''(x-a_j) \nonumber \\
&\; -\sum_{j=1}^{\cN}\sum_{k \neq j}^{\cN} r_jr_k \{\mP_j,\mP_k\} \alpha(a_j-a_k+\ii(r_j-r_k)\delta/2)\cA_{r_j}'(x-a_j) \nonumber  \\
&\; -\sum_{j=1}^{\cN} \mP_j^2 \cF_{r_j}'(x-a_j)        -\frac12 \sum_{j=1}^{\cN}\sum_{k\neq j}^{\cN} r_j r_k \{\mP_j,\mP_k\}\cF_{r_k}'(x-a_k).
\end{align}

To compute terms involving $\cT$, we need the following lemma.

\begin{lemma}\label{lem:cT}
The operator $\cT$ defined in \eqref{eq:cT} has the following action on the functions $\cA_{r_j}'(x-a_j)$
\begin{equation}\label{eq:cTApmprime}
(\cT \cA_{r_j}'(\cdot-a_j))(x)= \ii r_j \cA_{r_j}'(x-a_j) 
+(1+r_j)\ii \gamma_0 \left(\begin{array}{c} 0 \\ 1 \end{array}\right)+(1-r_j)\ii \gamma_0\left(\begin{array}{c} 1 \\ 0 \end{array}\right)
\quad (j=1,\ldots,\cN)
\end{equation}
provided \eqref{eq:ajbj} holds. 
\end{lemma}
\begin{proof}

Using the definitions \eqref{eq:EA} of $\cA_{\pm}(z)$ and recalling that $\zeta_2(z)=\zeta_1(z)+\gamma_0 z$ with $\gamma_0$ given by \eqref{eq:cc}, we write
\begin{equation}\label{eq:Apmalt}
\cA_{r_j}(x-a_j)=\left(\begin{array}{c}
+\zeta_1(x-a_j-\ii r_j\delta/2) \\ -\zeta_1(x-a_j+\ii r_j \delta/2)
\end{array}	\right) +\gamma_0 \left(\begin{array}{c} x-a_j-\ii r_j\delta/2 \\ -(x-a_j+ \ii r_j\delta/2)\end{array}\right).
\end{equation}
By differentiating \eqref{eq:Apmalt} with respect to $x$, we obtain
\begin{equation}\label{eq:Arjprime}
\cA_{r_j}'(x-a_j)=\left(\begin{array}{c}
-\wp_1(x-a_j-\ii r_j\delta/2) \\ +\wp_1(x-a_j+\ii r_j \delta/2)
\end{array}	\right) +\gamma_0 \left(\begin{array}{c} 1 \\ -1\end{array}\right),
\end{equation}
where $\wp_1(z)\coloneqq -\zeta_1'(z)$ is a $2\ell$-periodic, zero-mean function (see Appendix \ref{app:elliptic}). We compute the action of $\cT$ on the first and second terms in \eqref{eq:Arjprime} separately.

First, we use the following result \cite[Appendix~2.a]{berntson2020}: for a $2\ell$-periodic, zero-mean function $f(z)$ analytic in a strip $-A<\im(z)<A$ for some $A>\delta/2$, the vector-valued functions $(f(x\mp\ii\delta/2,-f(x\pm\ii\delta/2))^T$
are eigenfunctions of the $\cT$ operator with eigenvalues $\pm \ii$. Applied to the functions $f(z)=\wp_1(z-a_j)$, this gives
\begin{align}\label{eq:cTwp1}
\bigg(\cT\left(\begin{array}{c}
-\wp_1(\cdot-a_j-\ii r_j\delta/2) \\ +\wp_1(\cdot-a_j+\ii r_j \delta/2)
\end{array}	\right)\bigg)(x)
=&\; \ii r_j\left(\begin{array}{c}
-\wp_1(x-a_j-\ii r_j\delta/2) \\ +\wp_1(x-a_j+\ii r_j \delta/2)
\end{array}	\right) \nonumber \\
=&\; \ii r_j\cA_{r_j}'(x-a_j)-\ii r_j \gamma_0 \left(\begin{array}{c} 1\\ -1 \end{array}\right),	
\end{align}
where we have used \eqref{eq:Arjprime} in the second step.

Second, we recall from \cite[Appendix~B]{berntsonlangmann2021} that
\begin{equation}
T(1)=0,\qquad \tilde{T}(1)=-\ii 	
\end{equation}
and hence,
\begin{equation}\label{eq:cTconstant}
\cT\left(\begin{array}{c} 1 \\ -1 \end{array}\right)= \ii \left(\begin{array}{c} 1 \\ 1 \end{array}\right). 
\end{equation}

By applying $\cT$ to \eqref{eq:Arjprime} and using  and \eqref{eq:cTwp1} and \eqref{eq:cTconstant}, we obtain \eqref{eq:cTApmprime}. 
\end{proof}

Using \eqref{eq:ansatzSH0}, Lemma \ref{lem:cT}, and \eqref{eq:constraintPjQj} in the form 
\begin{equation}\label{eq:sumconstraint4}
\sum_{j=1}^{\cN} r_j \mP_j=0,
\end{equation}
we obtain
\begin{equation}\label{eq:TUx}
\cT\cU_{0,x}=-\sum_{j=1}^{\cN} \mP_j \big(\cA_{r_j}'(x-a_j) - \gamma_0\cE\big),\qquad \cT\cU_{0,xx}=-\sum_{j=1}^{\cN} \mP_j \cA_{r_j}''(x-a_j),
\end{equation}
where we have used the fact that $\mathcal{T}$ commutes with differentiation \cite{berntsonlangmann2021} to obtain the second identity.

From \eqref{eq:ansatzSH0} and the first equation in \eqref{eq:TUx}, we compute
\begin{align}\label{eq:UTUx1}
\ii[\cU_0\ocomma \cT\cU_{0,x}]=&\; -\ii \sum_{j=1}^{\cN} [\mM,\mP_j]\cE\circ \cA_{r_j}'(x-a_j)+\sum_{j=1}^{\cN}\sum_{k\neq j}^{\cN} r_j [\mP_j,\mP_k]\cA_{r_j}(x-a_j)\circ \cA_{r_k}'(x-a_k) \nonumber \\
&\; +\ii \gamma_0 [\cU_0\ocomma (\mP+\mQ)\cE]     \nonumber \\
=&\; -\ii \sum_{j=1}^{\cN} [\mM,\mP_j] \cA_{r_j}'(x-a_j)-\sum_{j=1}^{\cN}\sum_{k\neq j}^{\cN} r_j [\mP_j,\mP_k]\alpha(a_j-a_k+\ii(r_j-r_k)\delta/2)\cA_{r_k}'(x-a_k) \nonumber \\
&\; -\sum_{j=1}^{\cN}\sum_{k\neq j}^{\cN} r_j [\mP_j,\mP_k]V(a_j-a_k+\ii(r_j-r_k)\delta/2)\big(\cA_{r_j}(x-a_j)-\cA_{r_k}(x-a_k)\big) \nonumber \\
&\; +\frac12 \sum_{j=1}^{\cN}\sum_{k\neq j}^{\cN} r_j [\mP_j,\mP_k]\cF_{r_k}'(x-a_k)
 + \frac12\sum_{j=1}^{\cN}\sum_{k\neq j}^{\cN}r_j [\mP_j,\mP_k] \varkappa'(a_j-a_k+\ii(r_j-r_k)\delta/2)\cE \nonumber \\
 &\; -\ii \gamma_0 [ (\mP+\mQ)\cE\ocomma\cU_0],
\end{align}
where we have employed \eqref{eq:Id4} and used the anti-symmetry of the generalized commutator \eqref{eq:commutators} in the second step. Since $V(z)$ is an even function, we can rewrite the double sum in the second line of the right-hand side as follows:
\begin{multline}
\sum_{j=1}^{\cN}\sum_{k\neq j}^{\cN} r_j [\mP_j,\mP_k]V(a_j-a_k+\ii(r_j-r_k)\delta/2)\big(\cA_{r_j}(x-a_j)-\cA_{r_k}(x-a_k)\big) \\
=\frac12\sum_{j=1}^{\cN}\sum_{k\neq j}^{\cN} (r_j+r_k) [\mP_j,\mP_k]V(a_j-a_k+\ii(r_j-r_k)\delta/2)\big(\cA_{r_j}(x-a_j)-\cA_{r_k}(x-a_k)\big) \\
=\sum_{j=1}^{\cN}\sum_{k\neq j}^{\cN}( r_j+r_k )[\mP_j,\mP_k]V(a_j-a_k+\ii(r_j-r_k)\delta/2)\cA_{r_j}(x-a_j).
\end{multline}
Moreover, changing variables $j\leftrightarrow k$ in the double sum in the first line and in the first double sum in the third line of the right-hand side of \eqref{eq:UTUx1}, we arrive at
\begin{align}\label{eq:UTUx2}
\ii[\cU_0\ocomma \cT\cU_{0,x}]=&\; -\ii \sum_{j=1}^{\cN} [\mM,\mP_j] \cA_{r_j}'(x-a_j) - \sum_{j=1}^{\cN}\sum_{k\neq j}^{\cN} r_k [\mP_j,\mP_k]\alpha(a_j-a_k+\ii(r_j-r_k)\delta/2)\cA_{r_j}'(x-a_j) \nonumber \\
&\; -\sum_{j=1}^{\cN}\sum_{k\neq j}^{\cN} (r_j+r_k) [\mP_j,\mP_k]V(a_j-a_k+\ii(r_j-r_k)\delta/2)\cA_{r_j}(x-a_j) \nonumber \\
&\; -\frac12 \sum_{j=1}^{\cN}\sum_{k\neq j}^{\cN} r_k [\mP_j,\mP_k]\cF_{r_j}'(x-a_j)
 + \frac12\sum_{j=1}^{\cN}\sum_{k\neq j}^{\cN}r_j [\mP_j,\mP_k] \varkappa'(a_j-a_k+\ii(r_j-r_k)\delta/2)\cE  \nonumber \\
 &\; -\ii \gamma_0 [ (\mP+\mQ)\cE\ocomma\cU_0]. 
\end{align}

Inserting \eqref{eq:Ut}, \eqref{eq:UUx2}, the second equation in \eqref{eq:TUx}, and \eqref{eq:UTUx2} into \eqref{eq:sncILW3} gives
\begin{align}\label{eq:combined} 
0=&\; \Bigg( \dot{\mM}  +\frac14\sum_{j=1}^{\cN}\sum_{k\neq j}^{\cN} (r_j+r_k)[\mP_j,\mP_k]\varkappa'(a_j-a_k+\ii(r_j-r_k)\delta/2)         \Bigg)\cE \nonumber \\
&\; +\sum_{j=1}^{\cN} \Bigg( \ii r_j \dot{\mP}_j  - \sum_{k\neq j}^{\cN} (r_j+r_k)[\mP_j,\mP_k]V(a_j-a_k+\ii(r_j-r_k)\delta/2)\Bigg)\cA_{r_j}(x-a_j) \nonumber \\
&\; + \sum_{j=1}^{\cN} \Bigg( -\ii r_j\mP_j\dot{a}_j+\ii r_j\{\mM,\mP_j\} -\ii[\mM,\mP_j] \nonumber \\
&\; \phantom{+ \sum_{j=1}^{\cN} \Bigg( }-\sum_{k\neq j}^{\cN} r_k\big(r_j \{\mP_j,\mP_k\} +[\mP_j,\mP_k]  \big)  \alpha(a_j-a_k+\ii(r_j-r_k)\delta/2)   \Bigg) \cA_{r_j}'(x-a_j) \nonumber \\
&\; + \sum_{j=1}^{\cN}  \big( \mP_j^2-\mP_j   \big)    \cA_{r_j}''(x-a_j) -\frac12\sum_{j=1}^{\cN} \Bigg(2\mP_j^2+ \sum_{k\neq j}^{\cN} \big(r_jr_k\{ \mP_j, \mP_k\} + r_k[\mP_j, \mP_k] \big) \Bigg)\cF_{r_j}'(x-a_j),
\end{align}
where we have symmetrized the double sum in the first line using the fact that $\varkappa'(z)$ is an odd function. 

We first consider the conditions under which the terms proportional to $\cF_{r_j}'(x-a_j)$ vanish. We write
\begin{align}
2\mP_j^2+ \sum_{k\neq j}^{\cN} \big(r_jr_k\{ \mP_j, \mP_k\} + r_k[\mP_j, \mP_k] \big)  \nonumber 
=&\;  2\mP_j^2 -(1-r_j) \mX_j \mP_j + (1+r_j)\mP_j\mX_j,
\end{align}
where $\mX_j\coloneqq \sum_{k\neq j}^{\cN} r_k \mP_k$, and notice that the right-hand side vanishes if $\mX_j=-r_j\mP_j$, i.e., if \eqref{eq:sumconstraint4} holds. 

Thus, the functions $\mU$ and $\mV$ defined in \eqref{eq:ansatzgen} satisfy the sncILW equation whenever \eqref{eq:sumconstraint4} holds and the following conditions from \eqref{eq:combined} are fulfilled,
\begin{equation}\label{eq:dotMSH}
\dot{\mM}= -\frac14\sum_{j=1}^{\cN}\sum_{k\neq j}^{\cN} (r_j+r_k)[\mP_j,\mP_k]\varkappa'(a_j-a_k+\ii(r_j-r_k)\delta/2)
\end{equation}
and
\begin{align}
r_j\mP_j\dot{a}_j =&\; r_j\{\mM,\mP_j\}-[\mM,\mP_j]+\ii \sum_{k\neq j}^{\cN} r_k\big(r_j \{\mP_j,\mP_k\} + [\mP_j,\mP_k]  \big)  \alpha(a_j-a_k+\ii(r_j-r_k)\delta/2), \label{eq:1storder1}\\
r_j\dot{\mP}_j = &\;  -\ii\sum_{k\neq j}^{\cN} (r_j+r_k)[\mP_j,\mP_k]V(a_j-a_k+\ii(r_j-r_k)\delta/2), \\
\mP_j^2=&\; \mP_j,
\end{align}
for $j=1,\ldots,\cN$ and \eqref{eq:sumconstraint4}. Recalling the notation \eqref{eq:PjQj} and \eqref{eq:shorthandsolitons}, we see that these are equivalent to \eqref{eq:Mdotgen}, \eqref{eq:ajdotsolgen}--\eqref{eq:Pj2solgen} and \eqref{eq:constraintPjQj}. The result follows. 
\end{proof}

\section{B\"{a}cklund transformation}\label{sec:BT}
Throughout this section, $V(z)=\wp_2(z)$, $\alpha(z)=\zeta_2(z)$,  and $M=N$.

In this section, we prove that solutions of the elliptic sCM equations of motion \eqref{eq:sCMa}--\eqref{eq:sCMgh} 
are, under certain conditions, also solutions of a B\"{a}cklund transformation for the elliptic sCM system. This B\"acklund transformation is given by
\begin{equation}\label{eq:BT1}
\begin{split}
&\dot{a}_j \langle f_j|  = 2\langle f_j| \mM   +2\ii \sum_{k\neq j}^N \langle f_j|e_k\rangle\langle f_k|  \alpha(a_j-a_k) -2\ii \sum_{k=1}^N \langle f_j | g_k\rangle \langle h_k|  \alpha(a_j-b_k),     \\
&\dot{b}_j |g_j\rangle =  2\mM |g_j\rangle -2\ii \sum_{k\neq j}^N |g_k\rangle\langle h_k|g_j\rangle  \alpha(b_j-b_k) + 2\ii \sum_{k=1}^N |e_k\rangle\langle f_k | g_j\rangle \alpha(b_j-a_k)
\end{split}\quad (j=1,\ldots,N),
\end{equation}
together with \eqref{eq:sCMef}, \eqref{eq:sCMgh}, and \eqref{eq:Mdotgen}. The precise statement is as follows.

\begin{proposition}\label{prop:equiv}
Let $\mM$, $\{a_j,|e_j\rangle,\langle f_j|\}_{j=1}^{N}$, and $\{b_j,|g_j\rangle,\langle h_j|\}_{j=1}^N$ be a solution of the system consisting of \eqref{eq:Mdotgen} and the sCM systems \eqref{eq:sCMa}--\eqref{eq:sCMgh} 
on an interval $[0,\tau)$ for some $\tau\in (0,\infty)\cup\{\infty\}$, with initial conditions satisfying \eqref{eq:sCMconstraint1}, \eqref{eq:BT}, and \eqref{eq:constraintPjQj} at $t=0$, and where \eqref{eq:ajak2} holds on $[0,\tau)$. Then, the first-order system consisting of \eqref{eq:sCMef}, \eqref{eq:sCMgh}, \eqref{eq:Mdotgen}, and \eqref{eq:BT1} is satisfied on $[0,\tau)$.  
\end{proposition}

In order to prove Proposition~\ref{prop:equiv}, we need two lemmas. The first lemma shows that the first-order system of Proposition~\ref{prop:solitonsgen} admits a unique solution under mild assumptions, generalizing a known result in Cases I--III \cite{berntson2022consistency}.

\begin{lemma}\label{lem:consistency}
The initial value problem consisting of \eqref{eq:sCMef}, \eqref{eq:sCMgh}, \eqref{eq:Mdotgen}, and \eqref{eq:BT1} with initial conditions satisfying \eqref{eq:sCMconstraint1}, \eqref{eq:constraintPjQj}, and \eqref{eq:BT1} at $t=0$ has a unique solution on a maximal interval $[0,\tau_{\mathrm{max}})$ for some $\tau_{\mathrm{max}}\in (0,\infty)\cup\{\infty\}$. On this interval, it holds that
\begin{equation}\label{eq:ajak2}
\begin{gathered}
a_j\neq a_k \bmod \{2\ell,2\ii\delta\}, \quad b_j \neq b_k \bmod \{2\ell,2\ii\delta\} \quad (1\leq j<k\leq N), \\
a_j \neq b_k \bmod \{2\ell,2\ii\delta\}  \quad (j,k=1,\ldots,N).
\end{gathered}
\end{equation}
\end{lemma}

\begin{proof}
See Appendix~\ref{app:lemconsistency}. 	
\end{proof}

The second lemma shows that the first-order system of Proposition~\ref{prop:equiv} implies the second-order sCM equations of motion \eqref{eq:sCMa}--\eqref{eq:sCMgh}.

\begin{lemma}\label{lem:BT}
The solution of \eqref{eq:sCMef}, \eqref{eq:sCMgh}, \eqref{eq:Mdotgen}, and \eqref{eq:BT1} in Lemma~\ref{lem:consistency} solves \eqref{eq:sCMa} and \eqref{eq:sCMb} on $[0,\tau_{\mathrm{max}})$. 
\end{lemma}

\begin{proof}
See Appendix~\ref{app:lemBT}. 	
\end{proof}

\begin{proof}[Proof of Proposition~\ref{prop:equiv}]
By standard arguments (see, e.g., \cite[Theorem~8.1]{hartman1982}) the solution is unique on $[0,\tau)$. We show that this solution coincides with the solution of the initial value problem of Lemma~\ref{lem:consistency} on $[0,\tau)$, assuming compatible initial conditions.

First suppose, seeking a contradiction, that $\tau>\tau_{\mathrm{max}}$. Then, by Lemma~\ref{lem:BT} and standard arguments (see, e.g., \cite[Corollary~3.2]{hartman1982}), the solution of the initial value problem of Lemma~\ref{lem:consistency} solves \eqref{eq:sCMa}--\eqref{eq:sCMgh} and \eqref{eq:Mdotgen}, 
and either (i) one of the solution variables tends to infinity or (ii) \eqref{eq:ajak} is violated as $t\to \tau_{\mathrm{max}}$. Because our solution of the initial value problem associated with \eqref{eq:sCMa}--\eqref{eq:sCMgh} and \eqref{eq:Mdotgen} exists and is unique on $[0,\tau)$, possibility (i) is excluded.  By the assumption that \eqref{eq:ajak2} holds on $[0,\tau)$, possibility (ii) is also excluded and we have obtained a contradiction.  

Hence, $\tau\leq \tau_{\mathrm{max}}$. By the uniqueness of the solution to the initial value problem of Lemma~\ref{lem:consistency} on $[0,\tau_{\mathrm{max}})$, Lemma~\ref{lem:BT}, and the uniqueness of the solution of the initial value problem associated with \eqref{eq:sCMa}--\eqref{eq:sCMgh} and \eqref{eq:Mdotgen} on $[0,\tau)$, the result follows. 
\end{proof}

\section{Proof of Theorem~\ref{thm:sncILW} in Case IV}\label{section5}
Throughout this section, $V(z)=\wp_2(z)$, $\alpha(z)=\zeta_2(z)$,  and $M=N$.
We first establish the following lemma.

\begin{lemma}\label{lem:ajtbjt}
Let $\tilde{a}_j\coloneqq a_j-\ii\delta/2$ and $\tilde{b}_j\coloneqq b_j+\ii\delta/2$ for $j=1,\ldots,N$. Suppose that $\mM$, $\{\tilde{a}_j,|e_j\rangle,\langle f_j|\}_{j=1}^N$, and $\{\tilde{b}_j,|g_j\rangle,\langle h_j|\}_{j=1}^N$ satisfy the equations \eqref{eq:sCMef}, \eqref{eq:sCMgh}, \eqref{eq:Mdotgen}, and \eqref{eq:BT1} on $[0,\tau)$ with initial conditions satisfying \eqref{eq:sCMconstraint1} and \eqref{eq:constraintPjQj} at $t=0$. Then, \eqref{eq:ajdotsolgen}--\eqref{eq:Pj2solgen} and \eqref{eq:constraintPjQj} hold on $[0,\tau)$.
\end{lemma}

\begin{proof}
Consider the set of equations \eqref{eq:BT1} with the replacements
\begin{equation}\label{eq:ajtoatj}
\{a_j\}_{j=1}^N\to \{\tilde{a}_j\}_{j=1}^N,\qquad \{b_j\}_{j=1}^N\to \{\tilde{b}_j\}_{j=1}^N.	
\end{equation}
By left-multiplying the first set of the resulting equations by $|e_j\rangle$ and right-multiplying the second set of the resulting equations by $\langle h_j|$, we obtain 
\begin{equation}	\label{eq:BTmatrix}
\begin{split}	
\dot{\tilde{a}}_j |e_j\rangle\langle f_j|  =&\; 2|e_j\rangle\langle f_j| \mM   +2\ii \sum_{k\neq j}^N |e_j\rangle\langle f_j|e_k\rangle\langle f_k|  \alpha(\tilde{a}_j-\tilde{a}_k) -2\ii \sum_{k=1}^N |e_j\rangle\langle f_j | g_k\rangle \langle h_k|  \alpha(\tilde{a}_j-\tilde{b}_k),\\
\dot{\tilde{b}}_j |g_j\rangle\langle h_j| =&\;   2\mM |g_j\rangle\langle h_j| -2\ii \sum_{k\neq j}^N |g_k\rangle\langle h_k|g_j\rangle\langle h_j| \alpha(\tilde{b}_j-\tilde{b}_k) + 2\ii \sum_{k=1}^N |e_k\rangle\langle f_k | g_j\rangle\langle h_j| \alpha(\tilde{b}_j-\tilde{a}_k)
\end{split}
\end{equation}
for $j=1,\ldots,N$. Recalling the definitions \eqref{eq:PjQj} of $\mP_j$ and $\mQ_j$ and those of $\tilde{a}_j$ and $\tilde{b}_j$, we see that \eqref{eq:BTmatrix} is equivalent to \eqref{eq:ajdotsolgen} which thus holds on $[0,\tau)$.

By differentiating \eqref{eq:PjQj} with respect to time and inserting \eqref{eq:sCMef} and \eqref{eq:sCMgh} with \eqref{eq:ajtoatj}, we find
\begin{subequations}\label{eq:PjdotQjdot} 
\begin{equation}
\dot{\mP}_j= |\dot{e}_j\rangle\langle f_j|+|e_j\rangle\langle \dot{f}_j|= 2\ii \sum_{k\neq j}^N |e_k\rangle\langle f_k|e_j\rangle\langle f_j|V(\tilde{a}_j-\tilde{a}_k)-2\ii\sum_{k\neq j}^{N}|e_j\rangle\langle f_j|e_k\rangle\langle f_k|V(\tilde{a}_j-\tilde{a}_k) 
\end{equation}
and
\begin{equation}
\dot{\mQ}_j= |\dot{g}_j\rangle\langle h_j|+|g_j\rangle\langle \dot{h}_j|
= 2\ii \sum_{k\neq j}^N |g_k\rangle\langle h_k|g_j\rangle\langle h_j|V(\tilde{b}_j-\tilde{b}_k)-2\ii\sum_{k\neq j}^{N}|g_j\rangle\langle h_j|g_k\rangle\langle h_k|V(\tilde{b}_j-\tilde{b}_k) 
\end{equation}
for $j=1,\ldots,N$.
\end{subequations}
Recalling the definitions of $\mP_j$ and $\mQ_j$ in \eqref{eq:PjQj} and those of $\tilde{a}_j$ and $\tilde{b}_j$ we see that \eqref{eq:PjdotQjdot} is equivalent to \eqref{eq:Pjdotsolgen} which thus holds on $[0,\tau)$.

By differentiating the quantities $\langle f_j|e_j\rangle$ and $\langle h_j|g_j\rangle$ with respect to time and inserting \eqref{eq:sCMef} and \eqref{eq:sCMgh} with \eqref{eq:ajtoatj}, we find
\begin{subequations}\label{eq:trPjdot} 
\begin{align}
\frac{\mathrm{d}}{\mathrm{d}t}\langle f_j|e_j\rangle=&\; \langle \dot{f}_j|e_j\rangle+\langle f_j|\dot{e}_j\rangle  \nonumber \\
=&\;  -2\ii \sum_{k\neq j}^N \langle f_j|e_k\rangle\langle f_k| e_j\rangle V(\tilde{a}_j-\tilde{a}_k)+2\ii\sum_{k\neq j}^{N} \langle f_j|e_k\rangle\langle f_k|e_j\rangle   V(\tilde{a}_j-\tilde{a}_k) =0
\end{align}
and
\begin{align}
\frac{\mathrm{d}}{\mathrm{d}t}\langle h_j|g_j\rangle=&\; \langle \dot{h}_j|g_j\rangle+\langle h_j|\dot{g}_j\rangle  \nonumber \\
=&\;  -2\ii \sum_{k\neq j}^N \langle h_j|g_k\rangle\langle h_k| g_j\rangle V(\tilde{b}_j-\tilde{b}_k)+2\ii\sum_{k\neq j}^{N} \langle h_j|g_k\rangle\langle h_k|g_j\rangle   V(\tilde{b}_j-\tilde{b}_k) =0
\end{align}
for $j=1,\ldots,N$.
\end{subequations}
Because \eqref{eq:sCMconstraint1} holds at $t=0$ by assumption, \eqref{eq:trPjdot} guarantees it holds on $[0,\tau)$. Thus, by writing $\mP_j^2=|e_j\rangle\langle f_j |e_j\rangle\langle f_j|=\langle f_j|e_j\rangle\mP_j$ and $\mQ_j^2=|g_j\rangle\langle h_j |g_j\rangle\langle h_j|=\langle h_j|g_j\rangle \mQ_j$, we see that \eqref{eq:Pj2solgen} holds on $[0,\tau)$. 

To show that \eqref{eq:constraintPjQj} holds on $[0,\tau)$, we use \eqref{eq:Pj-Qjdot} with \eqref{eq:ajtoatj}; hence $\dot{\mP}=\dot{\mQ}=0=\dot{\mP}-\dot{\mQ}$ which implies that \eqref{eq:constraintPjQj} holds on $[0,\tau)$ because it holds at $t=0$. 
\end{proof}

We are now ready to prove Theorem~\ref{thm:sncILW}.

\begin{proof}[Proof of Theorem \ref{thm:sncILW}]
Suppose that $\mM(t)$, $\{a_j(t),|e_j(t)\rangle,\langle f_j(t)|\}_{j=1}^N$, and $\{b_j(t),|g_j(t)\rangle,\langle h_j(t)|\}_{j=1}^N$ defined for $t \in [0, \tau)$ satisfy the assumptions in the statement of Theorem \ref{thm:sncILW} for some $\tau > 0$. For $j=1,\ldots,N$, let $\tilde{a}_j(t) \coloneqq a_j(t)-\ii\delta/2$ and $\tilde{b}_j(t)\coloneqq b_j(t)+\ii\delta/2$. 
By assumption, $\mM$, $\{a_j,|e_j\rangle,\langle f_j|\}_{j=1}^N$, and $\{b_j,|g_j\rangle,\langle h_j|\}_{j=1}^N$ obey \eqref{eq:sCMa}--\eqref{eq:sCMgh} and \eqref{eq:Mdotgen} 
on $[0, \tau)$. The definitions of $\tilde{a}_j$ and $\tilde{b}_j$ imply that \eqref{eq:sCMa}--\eqref{eq:sCMgh}, 
and (\ref{eq:Mdotgen}) hold on $[0, \tau)$ also with $\{a_j, b_j\}_{j=1}^N$ replaced by $\{\tilde{a}_j, \tilde{b}_j\}_{j=1}^N$. 
Moreover, by assumption, the relations \eqref{eq:BT} hold at time $t = 0$. Using the $2\ii \delta$-periodicity of $\alpha$, it follows that the relations \eqref{eq:BT1} with $\{a_j, b_j\}_{j=1}^N$ replaced by $\{\tilde{a}_j, \tilde{b}_j\}_{j=1}^N$ hold at $t = 0$. We conclude that the functions $\mM$, $\{\tilde{a}_j,|e_j\rangle,\langle f_j|\}_{j=1}^N$, and $\{\tilde{b}_j,|g_j\rangle,\langle h_j|\}_{j=1}^N$ solve the initial value problem of Proposition~\ref{prop:equiv}, so by Proposition~\ref{prop:equiv} these functions satisfy the first-order system of Lemma~\ref{lem:consistency} on $[0,\tau)$. In other words, the equations \eqref{eq:BT1} with $\{a_j, b_j\}$ replaced by $\{\tilde{a}_j, \tilde{b}_j\}$ hold for all $t \in [0, \tau)$, so we can use Lemma \ref{lem:ajtbjt} to deduce that the functions $\mM$, $\{a_j,|e_j\rangle,\langle f_j|\}_{j=1}^N$, and $\{b_j,|g_j\rangle,\langle h_j|\}_{j=1}^N$ satisfy \eqref{eq:ajdotsolgen}--\eqref{eq:Pj2solgen} and \eqref{eq:constraintPjQj} on $[0,\tau)$. In particular, $\mM$, $\{a_j,|e_j\rangle,\langle f_j|\}_{j=1}^N$, and $\{b_j,|g_j\rangle,\langle h_j|\}_{j=1}^N$ fulfill the assumptions of Proposition \ref{prop:solitonsgen}. Thus we can employ Proposition \ref{prop:solitonsgen} to infer that the ansatz \eqref{eq:ansatzgen} provides a solution to the periodic sncILW equation \eqref{eq:sncILW}. This completes the proof of Theorem~\ref{thm:sncILW}.
\end{proof}

\section{Construction of soliton solutions} 
\label{sec:examples}

In this section, we show how to solve the nonlinear constraints on the initial data in Theorem~\ref{thm:sncILW}. While we specifically reference the constraints in Case IV, the following results can be straightforwardly adapted to Cases III in Theorem~\ref{thm:sncILW} and Cases I--II in Theorem~\ref{thm:sBO}. One-soliton solutions are derived in Section~\ref{subsec:one-soliton} and a linear algebra problem whose solutions parameterize the corresponding multi-soliton solutions is presented in Section~\ref{subsec:multi-solitons}

\subsection{One-soliton solution}\label{subsec:one-soliton}
Throughout this subsection, $\alpha(z)=\zeta_2(z)$.
For $N=M=1$, the time evolution equations \eqref{eq:sCMa}--\eqref{eq:sCMgh} and (\ref{eq:Mdotgen}) simplify to $\dot{\mM}=0$, 
\begin{equation}
\ddot{a}_1=0,\qquad |\dot{e}_1\rangle=0,\qquad \langle \dot{f}_1|=0,
\end{equation}
and
\begin{equation}
\ddot{b}_1=0,\qquad |\dot{g}_1\rangle=0,\qquad \langle \dot{h}_1|=0.
\end{equation}
Given initial conditions $\mM(0)=\mM_0$, 
\begin{equation}
a_1(0)=a_{1,0},\qquad      \dot{a}_1(0)=v_1,\qquad |e_1(0)\rangle=|e_{1,0}\rangle,\qquad \langle f_1(0)|=\langle f_{1,0}|,
\end{equation}
and
\begin{equation}
b_1(0)=b_{1,0},\qquad      \dot{b}_1(0)=w_1,\qquad |g_1(0)\rangle=|g_{1,0}\rangle,\qquad \langle h_1(0)|=\langle h_{1,0}|,
\end{equation}
we find the solution $\mM=\mM_0$, 
\begin{equation}
a_1=a_{1,0}+v_1 t,\qquad |e_1\rangle = |e_{1,0}\rangle,\qquad \langle f_1|=\langle f_{1,0}|,
\end{equation}
and
\begin{equation}
b_1=b_{1,0}+w_1 t,\qquad |g_1\rangle = |g_{1,0}\rangle,\qquad \langle h_1|=\langle h_{1,0}|.
\end{equation}

The constraints \eqref{eq:sCMconstraint1} and \eqref{eq:constraintPjQj} reduce to $\langle f_{1,0}|e_{1,0}\rangle=1=\langle h_{1,0}|g_{1,0}\rangle$ and $|e_{1,0}\rangle\langle f_{1,0}|=|g_{1,0}\rangle\langle h_{1,0}|$, respectively, which together imply $|g_{1,0}\rangle=c|e_{1,0}\rangle$ and $\langle h_{1,0}|=c^{-1} \langle f_{1,0}|$ for some $c\in\C\setminus\{0\}$. The constraint \eqref{eq:BT} then reduces to 
\begin{equation}\label{eq:M0ev}
\begin{split}
v_1 \langle f_{1,0}| =&\;  2\langle f_{1,0}| \mM_0-2\ii \langle f_{1,0} |\alpha(a_{1,0}-b_{1,0}+\ii\delta),    \\
w_1 |e_{1,0}\rangle = &\; 2\mM_0 |e_{1,0}\rangle - 2\ii  |e_{1,0}\rangle\alpha(a_{1,0}-b_{1,0}+\ii\delta),
\end{split}
\end{equation}
i.e., $\langle f_{1,0}|$ and $|e_{1,0}\rangle$ are left- and right-eigenvectors of $\mM_0$, respectively. By right-multiplying the first equation in \eqref{eq:M0ev} by $|e_{1,0}\rangle$ and left-multiplying the second equation in \eqref{eq:M0ev} by $\langle f_{1,0}|$, we find
\begin{equation}\label{eq:v1w1}
v_1=w_1=2\langle f_{1,0}|\mM_0|e_{1,0}\rangle -2\ii \alpha(a_{1,0}-b_{1,0}+\ii\delta).	
\end{equation}

Collecting the observations above and using Theorem~\ref{thm:sncILW}, we see that
\begin{equation}
\begin{split}
\mU(x,t)=&\; \ee^{2\ii\gamma_0 |e_{1,0}\rangle\langle f_{1,0}|t}	\mU_0(x,t)\ee^{-2\ii\gamma_0 |e_{1,0}\rangle\langle f_{1,0}|t},\\
\mV(x,t)=&\; \ee^{2\ii\gamma_0 |e_{1,0}\rangle\langle f_{1,0}|t}	\mV_0(x,t)\ee^{-2\ii\gamma_0 |e_{1,0}\rangle\langle f_{1,0}|t}	
\end{split}
\end{equation}
with 
\begin{equation}\label{eq:1sol}
\begin{split}
\mU_0(x,t)=&\;\mM_0+\ii |e_{1,0}\rangle\langle f_{1,0}|\big(\alpha(x-a_{1,0}-v_1 t-\ii\delta/2)-\alpha(x-b_{1,0}-v_1t+\ii\delta/2)\big), \\
\mV_0(x,t)=&\; -\mM_0-\ii |e_{1,0}\rangle\langle f_{1,0}|\big(\alpha(x-a_{1,0}-v_1 t+\ii\delta/2)-\alpha(x-b_{1,0}-v_1t-\ii\delta/2)\big), 	
\end{split}
\end{equation}
provides a solution of the sncILW equation \eqref{eq:sncILW} when $\langle f_{1,0}|$ and $|e_{1,0}\rangle$ are left- and right-eigenvectors, respectively of $\mM_0$ corresponding to the same eigenvalue and normalized to satisfy $\langle f_{1,0}|e_{1,0}\rangle=1$ and $v_1$ is given by \eqref{eq:v1w1}.

\begin{remark}
In the generic case where $v_1$ in \eqref{eq:v1w1} has a nonzero imaginary part, \eqref{eq:1sol} does not provide a traveling wave solution and \eqref{eq:ajbj} will be violated in finite time, after which Theorem~\ref{thm:sncILW} does not guarantee \eqref{eq:1sol} solves the sncILW equation \eqref{eq:sncILW}. For $v_1$ to be real, in which case \eqref{eq:1sol} provides a traveling wave solution of the sncILW equation \eqref{eq:sncILW} on $[0,\infty)$, it suffices for $\mM_0$ to be Hermitian (in this case $\langle f_{1,0}|=|e_{1,0}\rangle^{\dag}$ is a possibility but not a requirement unless all eigenvalues of $\mM_0$ are simple) and $b_{1,0}=a_{1,0}^*$. It is interesting to note that in the singular limit of the sncILW equation studied in \cite{berntsonklabbers2022}, no such one-soliton, traveling wave solutions exist. 
\end{remark}

\subsection{Solution of constraints: Case IV}\label{subsec:multi-solitons}
Throughout this subsection, $V(z)=\wp_2(z)$, $\alpha(z)=\zeta_2(z)$,  and $M=N$.
For each $j=1,\ldots,N$, let us identify the vector $|e_j\rangle\in \cV$ with the vector $\ve_j\in \C^d$ whose components $(\ve_j)_{\mu}$, $\mu=1,\ldots,d$ are the components of $|e_j\rangle$ with respect to some given basis of $\cV$. Next, let us identify the collection of $N$ vectors $\ve_j$, $j=1,\ldots,N$, with the single vector $\ve\in \C^{Nd}$ whose components $\ve_{j,\mu}\coloneqq (\ve_j)_{\mu}$ are indexed by $j=1,\ldots,N$ and $\mu=1,\ldots,d$. Similarly, let us identify the three collections of vectors $\{\langle f_j|\}_{j=1}^N$, $\{|g_j\rangle\}_{j=1}^N$, and $\{\langle h_j|\}_{j=1}^N$ with the vectors $\vf=(f_{j,\mu})\in \C^{Nd}$, $\vg=(g_{j,\mu})\in \C^{Nd}$, and $\vh=(h_{j,\mu})\in \C^{Md}$, respectively. Moreover, consider the matrix representation of $\mM$ with respect to the same basis of $\cV$. We identify $\mM$ with its vectorization $\vM\in \C^{d^2}$, i.e., the concatenation of the columns of $\mM$. Then, the constraints \eqref{eq:BT} and \eqref{eq:constraintPjQj} can be written as the $ (2N+d)d \times (2N+d)d$ linear system
\begin{equation}\label{eq:mAmBmC}
\left(\begin{array}{ccc}
\mA^1 & \mA^2 & \mA^3 \\ 
\mB^1 & \mB^2 & \mB^3 \\
\mC^1 & \mC^2 & 0
\end{array}\right)\left(\begin{array}{c}
\vh \\ \ve \\ \vM
\end{array}\right)
=\left(\begin{array}{c}
\mD^1 \vf \\ -\mD^2 \vg \\ \mathbf{0}
\end{array}\right),
\end{equation}
where $\mA^1\in \C^{Nd\times Nd}$, $\mA^2\in \C^{Nd\times Nd}$, $\mA^3\in \C^{Nd\times d^2}$, $\mB^1\in \C^{Nd\times Nd}$, $\mB^2\in \C^{Nd\times Nd}$, $\mB^3\in \C^{Nd\times d^2}$, $\mC^1\in \C^{d^2\times Nd}$, $\mC^2\in \C^{d^2\times Nd}$, $\mD^1\in\C^{Nd\times Nd}$, and $\mD^2\in\C^{Nd\times Nd}$ are defined by 
\begin{equation}\label{eq:mAmBmC2}
\begin{split}
& A^1_{j,\mu;k,\nu}= -2\ii \langle f_j|g_k\rangle \delta_{\mu,\nu}  \alpha(a_j-b_k+\ii\delta), \qquad A^2_{j,\mu;k,\nu}=  2\ii(1-\delta_{j,k}) f_{j,\nu}f_{k,\mu} \alpha(a_j-a_k) \\
&A^3_{j,\mu; \nu,\sigma}= 2f_{j,\sigma}\delta_{\mu,\nu}, \\
& B^1_{j,\mu;k,\nu}= 2\ii (1- \delta_{j,k}) g_{j,\nu}g_{k,\mu} \alpha(b_j-b_k), \qquad B^2_{j,\mu;k,\nu}= -2\ii \langle f_k|g_j\rangle \delta_{\mu,\nu} \alpha(b_j-a_k+\ii\delta),\\
&B^3_{j,\mu;\nu,\sigma}=  -2g_{j,\nu} \delta_{\mu,\sigma},     \\
& C^1_{\mu,\nu;j,\sigma}=    2g_{j,\nu} \delta_{\mu,\sigma}, \qquad C^2_{\mu,\nu;j,\sigma}=    -2f_{j,\mu}\delta_{\nu,\sigma},   \\
&D^1_{j,\mu;k,\nu}=  v_j \delta_{j,k}\delta_{\mu,\nu}, \qquad D^2_{j,\mu;k,\nu}=  w_j \delta_{j,k}\delta_{\mu,\nu}.
\end{split}
\end{equation}

We have observed in numerical experiments that the the square matrix in \eqref{eq:mAmBmC} is generically rank-deficient. Correspondingly, there are conditions on the vector on the right-hand side of \eqref{eq:mAmBmC} for the linear system to be consistent; if $\vf$ and $\vg$ are given, these are linear conditions on $\{v_j,w_j\}_{j=1}^N$. If these conditions are satisfied, the solution of \eqref{eq:mAmBmC} can be determined uniquely. The constraints \eqref{eq:sCMconstraint1} then lead to an overdetermined linear system of equations for the  remaining unknowns in $\{v_j,w_j\}_{j=1}^N$. However, in our numerical experiments, we have found that this system is uniquely solvable. 

\begin{remark}
Our conventions in this section differ slightly from those in \cite[Section~3.1.3]{berntsonlangmannlenells2022}, where Hermitian solutions of the sBO equation with $\mM=0$ are considered. There, bolded vectors are always identified with collections of kets, i.e., $\vf$ is identified with $\{|f_j\rangle\}_{j=1}^N$. We obtain Hermitian solutions of the sncILW equation from \eqref{eq:mAmBmC}--\eqref{eq:mAmBmC2} by setting $b_j=a_j^*$, $f_{j,\mu}=e_{j,\mu}^*$, and $h_{j,\mu}=g_{j,\mu}^*$ for $j=1,\ldots,N$ and $\mu=1,\ldots,d$. Note that in this case, the $(2N+d)d\times (2N+d)d$ matrix in \eqref{eq:mAmBmC} is itself Hermitian. 
\end{remark}

\subsection{Solution of constraints: Cases I-III}\label{subsec:multi-solitonsI--III}
In this subsection, $V(z)$ and $\alpha(z)$ are as in \eqref{eq:V} and \eqref{eq:alpha}  for the rational, trigonometric and hyperbolic cases.
The idea of the previous subsection can be adapted to Cases I--III. In these cases, the constraint \eqref{eq:PjQj} is not present, i.e., we may set $\mC_1=0$ and $\mC_2=0$. This yields an underdetermined system in the variables $\{\vh,\ve,\vM\}$ in \eqref{eq:mAmBmC}. To generate a (generically) consistent system, we rearrange \eqref{eq:mAmBmC} to 
\begin{equation}\label{eq:linearproblemI-III}
\left(\begin{array}{cc}
\mA^1 & \mA^2 \\ 
\mB^1 & \mB^2 	
\end{array}\right)\left(\begin{array}{c} \vh \\ \ve \end{array}\right)
=\left(\begin{array}{c} \mD^1\vf-\mA^3\vM \\ -\mD^2\vg-\mB^3\vM \end{array}\right).
\end{equation}
In Case III, the submatrices appearing in \eqref{eq:linearproblemI-III} are given by \eqref{eq:mAmBmC2} with $\alpha(z)$ in Case III \eqref{eq:alpha}. In Cases I--II, the submatrices appearing in \eqref{eq:linearproblemI-III} are given by \eqref{eq:mAmBmC2} with $\delta\to 0$ and $\alpha(z)$ in Cases I--II \eqref{eq:alpha}. Then, the method described in \cite[Section~3.1.3]{berntsonlangmannlenells2022} can be straightforwardly applied to the generate admissible initial data for Theorems~\ref{thm:sncILW} and \ref{eq:sBO}. 

\noindent
{\bf Acknowledgement} {\it We thank Patrick G\'{e}rard, Rob Klabbers, Enno Lenzmann, Masatoshi Noumi, Anton Ottosson, and Junichi Shiraishi for helpful discussions. We are grateful to Anton Ottosson for carefully reading the manuscript and suggesting improvements.
The work of B.K.B. was supported by the Olle Engkvist Byggm\"{a}stare Foundation, Grant 211-0122. 
E.L. gratefully acknowledges support from the European Research Council, Grant Agreement No.\ 2020-810451.
J.L. acknowledges support from the Ruth and Nils-Erik Stenb\"ack Foundation, the Swedish Research Council, Grant No.\ 2021-03877, and the European Research Council, Grant Agreement No. 682537.
}

\appendix

\section{Special functions}
\label{app:special} 
\subsection{Elliptic functions}
\label{app:elliptic}
We recall the standard definitions of the Weierstrass $\zeta$- and $\wp$-functions with half-periods $(\omega_1,\omega_2)$  \cite[Section 23.2]{DLMF}, 
\begin{equation} 
\zeta(z) \coloneqq \frac1{z} + \sum_{(n,m)\in \Z^2\setminus(0,0)}\left(\frac1{z-2n\omega_1-2m\omega_2} + \frac1{2n\omega_1+2m\omega_2} +  \frac{z}{(2n\omega_1+2m\omega_2)^2}  \right) 
\end{equation} 
and $\wp(z)=-\partial_z\zeta(z)$. The corresponding modified functions are 
\begin{equation} 
\zeta_j(z) \coloneqq \zeta(z)-\frac{\eta_j}{\omega_j}z,\quad \wp_j(z)\coloneqq -\partial_z\zeta_j(z) = \wp(x)+\frac{\eta_j}{\omega_j}\quad (j=1,2) 
\end{equation} 
where $\eta_j\coloneqq \zeta(\omega_j)$. Since $\zeta(z+2\omega_j)=\zeta(z)+2\eta_j$ for $j=1,2$ (see \cite[Eq.~(23.2.11)]{DLMF}), $\zeta_j(z+2\omega_j)=\zeta_j(z)$ for $j=1,2$. 
In the main text we have $(\omega_1,\omega_2)=(\ell,\ii\delta)$. 

Our definitions imply  
\begin{equation} 
\zeta_2(z)=\zeta_1(z)+\gamma_0 z
\end{equation} 
with $\gamma_0=\eta_1/\omega_1-\eta_2/\omega_2=(\eta_1\omega_2-\eta_2\omega_1)/(\omega_1\omega_2)$. Using  the well-known identity $\eta_1\omega_2-\eta_2\omega_1 = \frac12\pi\ii$ \cite[Eq.~(23.2.14)]{DLMF}, we obtain 
\begin{equation} 
\gamma_0 = \frac{\pi\ii}{2\omega_1\omega_2}, 
\end{equation} 
which for $(\omega_1,\omega_2)=(\ell,\ii\delta)$ gives \eqref{eq:cc}. 
 
\subsection{Functional identities}
\label{app:functional} 
We state and outline proofs of several well-known identities involving the special functions $\zeta_2(z)$, $\wp_2(z)$, and $\varkappa(z)$ that we use (more detailed proofs of these identities can be found in \cite{berntsonlangmann2021}, for example). Throughout this subsection, $z$ is a complex variable.

First, the functions $\zeta_2(z)$ and $\wp_2(z)$ are odd and even, respectively:
\begin{equation}\label{eq:parity}
\zeta_2(-z)=-\zeta_2(z),\qquad \wp_2(-z)=\wp_2(z) 
\end{equation}
(this is obvious from the definitions). 
Second,
\begin{equation}
\label{eq:Id1}
\wp_2(z)=-\partial_z\zeta_2(z)=\zeta_2(z)^2-\varkappa(z) 
\end{equation}
(this is implied by the definitions \eqref{eq:wp2f2}). 
Third, 
\begin{multline}
\label{eq:Id2}
\zeta_2(a-b)\zeta_2(b-c)+\zeta_2(b-c)\zeta_2(c-a)+\zeta_2(c-a)\zeta_2(a-b)\\
=-\frac12\big(\varkappa(a-b)+\varkappa(b-c)+\varkappa(c-a)\big)-\frac{3\eta_2}{2\ii\delta} \quad (a,b,c\in\C)
\end{multline}
(to get this, start with the well-known identity $\left( \zeta(x)+\zeta(y)+\zeta(z)\right)^2=\wp(x)+\wp(y)+\wp(z)$ for $x+y+z=0$, specialize to $(x,y,z)=(a-b,b-c,c-a)$, use definitions \eqref{eq:zetaj} and \eqref{eq:wp2f2} to write this as 
\begin{equation} 
\left( \zeta_2(a-b)+\zeta_2(b-c)+\zeta(c-a)\right)^2 = \wp_2(a-b)+\wp_2(b-c)+\wp_2(c-a)-\frac{3\eta_2}{\ii\delta}, 
\end{equation} 
and use the identity $\wp_2(z)=\zeta_2(z)^2-\varkappa(z)$
to obtain \eqref{eq:Id2}). 
Finally, the function $\zeta_2(z)$ is quasi-periodic with respect to $2\ell$ and periodic with respect to $2\ii\delta$:
\begin{equation}
\label{eq:zeta2periodic}
\zeta_2(z+2\ell)=\zeta_2(z)+\frac{\pi}{\delta},\qquad \zeta_2(z+2\ii\delta)=\zeta_2(z)  
\end{equation}
(this follows from the well-known identities $\zeta(z+2\omega_j)=\zeta(z)+2\eta_j$ for $j=1,2$ and the definition of $\zeta_2(z)$)
while $\wp_2(z)$ is $2\ell$- and $2\ii\delta$-periodic:
\begin{equation}\label{eq:wp2periodic}
\wp_2(z+2\ell)=\wp_2(z),\qquad \wp_2(z+2\ii\delta)=\wp_2(z). 
\end{equation}

 \section{Proofs}\label{app:Proofs}

 \subsection{Proof of Lemma~\ref{lem:consistency}}\label{app:lemconsistency}
 
 We use the shorthand notation \eqref{eq:shorthandsolitons} and 
\begin{equation}\label{eq:Bj}
 \mB_j\coloneqq \mM+\ii\sum_{k\neq j}^{\cN} r_k \mP_k\alpha(a_j-a_k) \quad (j=1,\ldots,\cN)
\end{equation}
to write \eqref{eq:BT1} as 
\begin{equation}
\label{eq:BTshort} 
\begin{split}  
\dot a_j\langle f_j| = &\;  2\langle f_j|\mB_j\quad (j=1,\ldots,N),      \\
\dot a_j|e_j\rangle = &\; 2\mB_j|e_j\rangle \quad (j=N+1,\ldots,\cN).
\end{split}
\end{equation} 

Moreover, \eqref{eq:Pjdotsolgen} becomes
\begin{equation}\label{eq:PjdotSH}
\dot{\mP}_j=-\ii\sum_{j=1}^{\cN} (1+r_jr_k)[\mP_j,\mP_k]V(a_j-a_k);	
\end{equation}
\eqref{eq:sCMef} and \eqref{eq:sCMgh} become
\begin{equation} 
\begin{split} 
\label{eq:sCM3b} 
|\dot e_j\rangle &= \ii\sum_{k\neq j}^\cN (1+r_jr_k) \mP_k|e_j\rangle V(a_j-a_k),\\
\langle\dot f_j| & = -\ii\sum_{k\neq j}^\cN(1+r_jr_k)\langle f_j|\mP_k V(a_j-a_k), 
\end{split} \quad (j=1,\ldots,\cN);
\end{equation}
\eqref{eq:sCMconstraint1} becomes
\begin{equation}\label{eq:efSH}
\langle e_j | f_j \rangle=1 \quad (j=1,\ldots,\cN);
\end{equation}
and \eqref{eq:Mdotgen} becomes
\begin{equation}\label{eq:MdotSH}
\dot{\mM}=-\frac14 \sum_{k=1}^{\cN} \sum_{l\neq k}^{\cN}(r_k+r_l) [\mP_k,\mP_l]\varkappa'(a_k-a_l).
\end{equation}

By right-multiplying the first set of equations in \eqref{eq:BTshort} by $|e_j\rangle$ and left-multiplying the second set of equations in \eqref{eq:BTshort} by $\langle f_j|$, we obtain
\begin{equation}\label{eq:BTscalar}
\begin{split}
\dot{a}_j=&\; 2\langle f_j|\mB_j| e_j\rangle \quad (j=1,\ldots,\cN).
\end{split}
\end{equation}
By the Picard-Lindel\"{o}f theorem, the system of equations consisting of \eqref{eq:sCM3b}, \eqref{eq:MdotSH}, and \eqref{eq:BTscalar} with the given initial data has a unique local solution. This solution may be extended as long as (i) no solution variable tends to infinity and (ii) \eqref{eq:ajak} holds (see, e.g., \cite[Corollary~3.2]{hartman1982}), up to a maximal time $\tau_{\mathrm{max}}\in (0,\infty)\cup\{\infty\}$. Moreover, this \textit{maximal solution} is unique on $[0,\tau_{\mathrm{max}})$ (see, e.g., \cite[Theorem~8.1]{hartman1982}). 

Let $\mM$ and $\{a_j,|e_j\rangle,\langle f_j|\}_{j=1}^{\cN}$ be this maximal solution. It follows from \eqref{eq:BTscalar} that
\begin{equation}\label{eq:ajPj3}
\begin{split}
\dot{a}_j\langle f_j|=&\; 2\langle f_j|\mB_j\mP_j \quad(j=1,\ldots,N), \\
\dot{a}_j|e_j\rangle =&\; 2\mP_j\mB_j|e_j\rangle \quad (j=N+1,\ldots,\cN).
\end{split}
\end{equation}
holds on $[0,\tau_{\mathrm{max}})$. The overdetermined system of equations for $\{a_j\}_{j=1}^{\cN}$ \eqref{eq:BTshort} will also be satisfied on $[0,\tau_{\mathrm{max}})$ if the difference between the right-hand sides of \eqref{eq:BTshort} and \eqref{eq:ajPj3} vanish on $[0,\tau_{\mathrm{max}})$. These differences are given by 
\begin{equation}\label{eq:FjEj} 
\begin{split}
\langle F_j|\coloneqq &\; \langle f_j|\mB_j-\langle f_j|\mB_j|e_j\rangle\langle f_j| =\langle f_j|\mB_j(1-\mP_j)\quad (j=1,\ldots,N), \\
|E_{j}\rangle \coloneqq &\;  \mB_j|e_j\rangle -|e_j\rangle\langle f_j|\mB_j|e_j\rangle=(1-\mP_j)\mB_j|e_j\rangle \quad (j=N+1,\ldots,\cN).
\end{split}	
\end{equation}

We will show that the time evolution of the quantities $\{\langle F_j|\}_{j=1}^N$ and $\{|E_{N+j}\rangle\}_{j=1}^N$ is determined by a linear, homogeneous (in both $\{\langle F_j|\}_{j=1}^N$ and $\{|E_{N+j}\rangle\}_{j=1}^N$) system of ordinary differential equations. The precise form of this system is not needed to establish our result, and so we introduce an equivalence relation $\simeq$ between two expressions that differ only by terms linear in $\{\langle F_j|\}_{j=1}^N$ and $\{|E_{N+j}\rangle\}_{j=1}^N$ with regular coefficients (the regularity of all such coefficients is guaranteed by conditions (i), (ii) in the discussion of maximal solutions above). Thus we only need to show that
\begin{equation}\label{eq:FjEjODEs}
\begin{split}
\langle \dot{F}_j|\simeq &\; 0	 \quad (j=1,\ldots,N), \\
|\dot{E}_{j}\rangle \simeq &\; 0 \quad (j=N+1,\ldots,\cN). 
\end{split}
\end{equation}
However, in the system of equations given by \eqref{eq:sCMef}, \eqref{eq:sCMgh}, \eqref{eq:Mdotgen}, \eqref{eq:sCMconstraint1}, \eqref{eq:constraintPjQj}, and \eqref{eq:BT1}, the variables $\{a_j,|e_j\rangle,\langle f_j|\}_{j=1}^N$ and $\{b_j,\langle h_j|, |g_j\rangle\}_{j=1}^N$ can be swapped by Hermitian conjugation and consequently, $\{\langle F_j|\}_{j=1}^N$ and $\{|E_{N+j}\rangle\}_{j=1}^N$ can be interchanged using this same symmetry. For this reason, it suffices to show that the first set of equations in \eqref{eq:FjEjODEs} holds.   

We differentiate the first equation in \eqref{eq:FjEj} with respect to time, which gives (for notational simplicity, we suppress the $j$-dependence of the quantities $\langle C_1|$, $\langle C_2|$, $\langle C_3|$)
\begin{equation}	
\langle \dot{F}_j|= \langle C_1|+\langle C_2|+\langle C_3|,
\end{equation}
where
\begin{equation}\label{eq:C1C2C3}
\begin{split}
\langle C_1|=&\; \langle \dot{f}_j|\mB_j(1-\mP_j), \\
\langle C_2|=&\; -\langle f_j|\mB_j\dot{\mP}_j, \\
\langle C_3|=&\; \langle f_j|\dot{\mB}_j(1-\mP_j). 
\end{split}	
\end{equation}
We compute each of these quantities in turn (note that $r_j=1$ in what follows). 

By inserting \eqref{eq:sCM3b} into $\langle C_1|$ in \eqref{eq:C1C2C3}, we find
\begin{align}\label{eq:C1}
\langle C_1|=&\; -\ii \sum_{k\neq j}^{\cN}(1+r_k)\langle f_j|\mP_k\mB_j(1-\mP_j) V(a_j-a_k)	
\end{align}
and, similarly, by inserting \eqref{eq:PjdotSH} into $\langle C_2|$ in \eqref{eq:C1C2C3}, we compute
\begin{align}\label{eq:C2}
\langle C_2|=&\; \ii\sum_{k\neq j}^{\cN}(1+r_k)\langle f_j|\mB_j [\mP_j,\mP_k]V(a_j-a_k) \nonumber \\
=&\; -\ii\sum_{k\neq j}^{\cN} (1+r_k)\langle f_j|\mB_j[1-\mP_j,\mP_k]V(a_j-a_k) \nonumber \\
=&\; -\ii\sum_{j=1}^{\cN}(1+r_k)\langle F_j|\mP_kV(a_j-a_k)+\ii\sum_{k\neq j}^{\cN}(1+r_k)\langle f_j|\mB_j\mP_k(1-\mP_j)V(a_j-a_k).
\end{align}
Then, from \eqref{eq:C1} and \eqref{eq:C2}, it follows that
\begin{align}\label{eq:C1C2sum}
\langle C_1|+\langle C_2|\simeq &\;  -\ii\sum_{k\neq j}^{\cN} (1+r_k)\langle f_j|[\mP_k,\mB_j](1-\mP_j)V(a_j-a_k). 
\end{align}

To compute $\langle C_3|$ in \eqref{eq:C1C2C3}, we differentiate the expression for $\mB_j$ in \eqref{eq:Bj} to write
\begin{equation}
\langle C_3|= \langle C_{3,1}|+\langle C_{3,2}|+\langle C_{3,3}|	,
\end{equation}
where
\begin{equation}\label{eq:C31C32C33}
\begin{split}
\langle C_{3,1}| \coloneqq &\; \langle f_j|\dot{\mM}(1-\mP_j), \\
\langle C_{3,2}| \coloneqq &\; \ii\sum_{k\neq j}^{\cN} r_k \langle f_j|\dot{\mP}_k(1-\mP_j) \alpha(a_j-a_k), \\
\langle C_{3,3}| \coloneqq &\; -\ii \sum_{k\neq j}^{\cN} r_k \langle f_j|\mP_k(1-\mP_j)(\dot{a}_j-\dot{a}_k)V(a_j-a_k).
\end{split}	
\end{equation}

Using \eqref{eq:MdotSH} in $\langle C_{3,1}|$ in \eqref{eq:C31C32C33} gives
\begin{align}\label{eq:C31}
\langle C_{3,1}|=&\; -\frac14\sum_{k=1}^{\cN}\sum_{l\neq k}^{\cN} (r_k+r_l)	 \langle f_j| [\mP_k,\mP_l](1-\mP_j)\varkappa'(a_k-a_l).  
\end{align}

Next, by inserting \eqref{eq:PjdotSH} into $\langle C_{3,2}|$ in \eqref{eq:C31C32C33}, we compute
\begin{align}\label{eq:C32}
\langle C_{3,2}|=&\; \sum_{k\neq j}^{\cN} \sum_{l\neq k}^{\cN} r_k(1+r_kr_l)\langle f_j|[\mP_k,\mP_l](1-\mP_j) \alpha(a_j-a_k)V(a_k-a_l) \nonumber\\
=&\; - \sum_{k\neq j}^{\cN} (1+r_k)\langle f_j|	[\mP_j,\mP_k](1-\mP_j)\alpha(a_j-a_k)V(a_j-a_k) \nonumber\\
&\; + \sum_{k\neq j}^{\cN}\sum_{l\neq j,k}^{\cN} (r_k+r_l)\langle f_j|[\mP_k,\mP_l](1-\mP_j)\alpha(a_j-a_k)V(a_k-a_l) \nonumber\\
=&\; -\sum_{k\neq j}^{\cN} (1+r_k)\langle f_j|\mP_j\mP_k(1-\mP_j)\alpha(a_j-a_k)V(a_j-a_k) \nonumber \\
&\; + \sum_{k\neq j}^{\cN}\sum_{l\neq j,k}^{\cN} (r_k+r_l)\langle f_j|[\mP_k,\mP_l](1-\mP_j)\alpha(a_j-a_k)V(a_k-a_l).
\end{align}
We rewrite the double sum as follows
\begin{multline}
	\sum_{k\neq j}^{\cN}\sum_{l\neq j,k}^{\cN} (r_k+r_l)\langle f_j|[\mP_k,\mP_l](1-\mP_j)\alpha(a_j-a_k)V(a_k-a_l) \\
	=\frac12\sum_{k\neq j}^{\cN}\sum_{l\neq j,k}^{\cN} (r_k+r_l)\langle f_j|[\mP_k,\mP_l](1-\mP_j)\big(\alpha(a_j-a_k)-\alpha(a_j-a_l)\big)V(a_k-a_l) \\
	=\sum_{k\neq j}^{\cN}\sum_{l\neq j,k}^{\cN} r_l\langle f_j|[\mP_k,\mP_l](1-\mP_j)\big(\alpha(a_j-a_k)-\alpha(a_j-a_l)\big)V(a_k-a_l);
\end{multline}
using also that $\langle f_j|\mP_j=\langle f_j|$ this gives
\begin{align}\label{eq:C32,2}
\langle C_{3,2}|=&\; -\sum_{k\neq j}^{\cN} (1+r_k)\langle f_j|\mP_k(1-\mP_j)\alpha(a_j-a_k)V(a_j-a_k) \nonumber \\
&\; +\sum_{k\neq j}^{\cN}\sum_{l\neq j,k}^{\cN} r_l\langle f_j|[\mP_k,\mP_l](1-\mP_j)\big(\alpha(a_j-a_k)-\alpha(a_j-a_l)\big)V(a_k-a_l).	
\end{align}

To compute $\langle C_{3,3}|$, we use the following relations, which follow from  \eqref{eq:BTscalar} and \eqref{eq:FjEj},
\begin{equation}\label{eq:BTEjFj}
\begin{split}	
\dot{a}_j\langle f_j|= &\; -2\langle F_j|+2\langle f_j|\mB_j \simeq 2\langle f_j|\mB_j  \quad (j=1,\ldots,N), \\	
\dot{a}_j |e_j\rangle= &\; -2|E_{j}\rangle+2\mB_j|e_j\rangle \simeq 2\mB_j|e_j\rangle   \quad (j=N+1,\ldots,\cN).
\end{split}	
\end{equation}
Note that \eqref{eq:BTEjFj} implies 
\begin{equation}\label{eq:BTEjFj2}
\dot{a}_j\mP_j\simeq (1+r_j)\mP_j\mB_j+(1-r_j)\mB_j\mP_j 	\quad (j=1,\ldots,\cN).
\end{equation}

Using \eqref{eq:BTEjFj} and \eqref{eq:BTEjFj2} in $\langle C_{3,3}|$ in \eqref{eq:C31C32C33}, we compute
\begin{align}\label{eq:C33}
\langle C_{3,3}| \simeq &\; -2\ii\sum_{k\neq j}^{\cN} r_k \langle f_j|\mB_j\mP_k(1-\mP_j)V(a_j-a_k) \nonumber \\
&\; +\ii \sum_{k\neq j}^{\cN} \langle f_j| \big((1+r_k)\mP_k\mB_k-(1-r_k)\mB_k\mP_k)(1-\mP_j)V(a_j-a_k) \nonumber \\
=&\; -2\ii\sum_{k\neq j}^{\cN} r_k\langle f_j|(\mB_j-\mB_k)\mP_k(1-\mP_j)V(a_j-a_k) \nonumber \\
&\; +\ii\sum_{k\neq j}^{\cN} (1+r_k)\langle f_j|[\mP_k,\mB_k](1-\mP_j)V(a_j-a_k). 
\end{align}

By combining \eqref{eq:C1C2sum}, \eqref{eq:C32,2}, and \eqref{eq:C33}, we arrive at
\begin{align}\label{eq:C1C2C32C33}
&\langle C_1|+\langle C_2|+\langle C_{3,2}|+\langle C_{3,3}| \nonumber \\
&\simeq  -\sum_{k\neq j}^{\cN} (1+r_k)\langle f_j|\mP_k(1-\mP_j)\alpha(a_j-a_k)V(a_j-a_k) \nonumber \\
&\;\phantom{\simeq} +\sum_{k\neq j}^{\cN}\sum_{l\neq j,k}^{\cN} r_l\langle f_j|[\mP_k,\mP_l](1-\mP_j)\big(\alpha(a_j-a_k)-\alpha(a_j-a_l)\big)V(a_k-a_l) \nonumber \\
&\;\phantom{\simeq}-2\ii\sum_{k\neq j}^{\cN} r_k\langle f_j|(\mB_j-\mB_k)\mP_k(1-\mP_j)V(a_j-a_k) \nonumber \\
&\;\phantom{\simeq} -\ii\sum_{k\neq j}^{\cN} (1+r_k)\langle f_j|[\mP_k,\mB_j-\mB_k](1-\mP_j)V(a_j-a_k). 
\end{align}

To proceed, we compute a convenient expression for $\mB_j-\mB_k$ directly from \eqref{eq:Bj},
\begin{align}\label{eq:BjBk}
\mB_j-\mB_k=&\; \ii \sum_{l\neq j}^{\cN}	 r_l\mP_l \alpha(a_j-a_l)-\ii \sum_{l\neq k}^{\cN}	 r_l\mP_l \alpha(a_k-a_l) \nonumber\\
=&\; \ii(r_k\mP_k+\mP_j)\alpha(a_j-a_k)+\ii\sum_{l\neq j,k}^{\cN} r_l\mP_l\big(\alpha(a_j-a_l)-\alpha(a_k-a_l)\big).
\end{align}

The final two lines of \eqref{eq:C1C2C32C33} can be written, using \eqref{eq:BjBk} and the relations $\langle f_j | \mP_j = \langle f_j|$ and $\mP_j(1 - \mP_j) = 0$, as
\begin{align}\label{eq:C1C2C32C33,2}
&2\sum_{k\neq j}^{\cN}(1+r_k) \langle f_j|\mP_k(1-\mP_j)\alpha(a_j-a_k)V(a_j-a_k)\nonumber \\
& +2\sum_{k\neq j}^{\cN}\sum_{l\neq j,k}^{\cN} r_kr_l\langle f_j|\mP_l\mP_k (1-\mP_j) \big(\alpha(a_j-a_l)-\alpha(a_k-a_l)\big)V(a_j-a_k) \nonumber\\
& - \sum_{k\neq j}^{\cN} (1+r_k)\langle f_j|\mP_k(1-\mP_j)\alpha(a_j-a_k)V(a_j-a_k) \nonumber\\
&+\sum_{k\neq j}^{\cN}\sum_{l\neq j,k}^{\cN} (1+r_k)r_l\langle f_j|[\mP_k,\mP_l](1-\mP_j)\big(\alpha(a_j-a_l)-\alpha(a_k-a_l)\big)V(a_j-a_k);
\end{align}
inserting this into \eqref{eq:C1C2C32C33} leads to cancellation of diagonal terms,
\begin{align}\label{eq:C1C2C32C33,3}
&\langle C_1|+\langle C_2|+\langle C_{3,2}|+\langle C_{3,3}| \nonumber \\
&\simeq\sum_{k\neq j}^{\cN}\sum_{l\neq j,k}^{\cN} r_l\langle f_j|[\mP_k,\mP_l](1-\mP_j)\big(\alpha(a_j-a_k)-\alpha(a_j-a_l)\big)V(a_k-a_l) \nonumber \\
&\phantom{\simeq} +\sum_{k\neq j}^{\cN}\sum_{l\neq j,k}^{\cN} (1+r_k)r_l\langle f_j|\mP_k\mP_l(1-\mP_j)\big(\alpha(a_j-a_l)-\alpha(a_k-a_l)\big)V(a_j-a_k) \nonumber \\
&\phantom{\simeq} -\sum_{k\neq j}^{\cN}\sum_{l\neq j,k}^{\cN} (1-r_k)r_l\langle f_j|\mP_l\mP_k(1-\mP_j)\big(\alpha(a_j-a_l)-\alpha(a_k-a_l)\big)V(a_j-a_k) \nonumber \\
&= \sum_{k\neq j}^{\cN}\sum_{l\neq j,k}^{\cN} r_l\langle f_j|[\mP_k,\mP_l](1-\mP_j) \nonumber \\
&\phantom{\sum_{k\neq j}^{\cN}\sum_{l\neq j,k}^{\cN}}\times\big(\big(\alpha(a_j-a_k)-\alpha(a_j-a_l)\big)V(a_k-a_l) +\big(\alpha(a_j-a_l)-\alpha(a_k-a_l)\big)V(a_j-a_k)\big) \nonumber \\
&\phantom{=}+\sum_{k\neq j}^{\cN}\sum_{l\neq j,k}^{\cN} r_kr_l \langle f_j|\{\mP_k,\mP_l\}(1-\mP_j)\big(\alpha(a_j-a_l)-\alpha(a_k-a_l)\big)V(a_j-a_k).
\end{align}

Using the identity
\begin{align}\label{eq:alphaId2}
\big(\alpha(a_j-a_l)-\alpha(a_k-a_l)\big)V(a_j-a_k)=&\; -\big(\alpha(a_j-a_k)-\alpha(a_j-a_l)\big)V(a_k-a_l) \nonumber \\
&\; -\frac12\big(\varkappa'(a_j-a_k)-\varkappa'(a_k-a_l)\big),
\end{align}
which can be obtained by differentiating \eqref{eq:Id2} with respect to $b$ and setting $a=a_j$, $b=a_k$, and $c=a_l$, in \eqref{eq:C1C2C32C33,3} gives
\begin{align}\label{eq:C1C2C32C33,4}
&\langle C_1|+\langle C_2|+\langle C_{3,2}|+\langle C_{3,3}| \nonumber\\
&\simeq -\frac12\sum_{k\neq j}^{\cN}\sum_{l\neq j,k}^{\cN} r_l\langle f_j|[\mP_k,\mP_l](1-\mP_j)\big(\varkappa'(a_j-a_k)-\varkappa'(a_k-a_l)\big) \nonumber\\
&\;\phantom{\simeq}-\sum_{k\neq j}^{\cN}\sum_{l\neq j,k}^{\cN} r_kr_l \langle f_j| \{\mP_k,\mP_l\}(1-\mP_j)\big(\alpha(a_j-a_k)-\alpha(a_j-a_l)\big)V(a_k-a_l) \nonumber\\
&\;\phantom{\simeq} -\frac12 \sum_{k\neq j}^{\cN}\sum_{l\neq j,k}^{\cN} r_kr_l\langle f_j|\{\mP_k,\mP_l\} (1-\mP_j)\big(\varkappa'(a_j-a_k)-\varkappa'(a_k-a_l)\big).
\end{align}

The second double sum in \eqref{eq:C1C2C32C33,4} vanishes by symmetry, as does the part of the third double sum proportional to $\varkappa'(a_k-a_l)$. We use
\begin{equation}\label{eq:sumconstraint4SH}
\sum_{l\neq j,k}^{\cN} r_l\mP_l=-\mP_j-r_k\mP_k,
\end{equation}
a consequence of \eqref{eq:sumconstraint4}, to simplify what remains. This allows us to compute
\begin{align}
&\langle C_1|+\langle C_2|+\langle C_{3,2}|+\langle C_{3,3}| \nonumber \\
&\simeq \frac12 \sum_{k\neq j}^{\cN} \langle f_j|[\mP_k, \mP_j+r_k\mP_k](1-\mP_j)\varkappa'(a_j-a_k)+\frac12\sum_{k\neq j}^{\cN}\sum_{l\neq j,k}^{\cN} r_l\langle f_j|[\mP_k,\mP_l](1-\mP_j)\varkappa'(a_k-a_l)	\nonumber \\
&\;\phantom{\simeq} +\frac12\sum_{k\neq j}^{\cN} r_k\langle f_j|\{\mP_k,\mP_j+r_k\mP_k\}(1-\mP_j)\varkappa'(a_j-a_k) \nonumber \\
&= -\frac12 \sum_{k\neq j}^{\cN}  \langle f_j|\mP_k(1-\mP_j)\varkappa'(a_j-a_k)+\frac12\sum_{k\neq j}^{\cN}\sum_{l\neq j,k}^{\cN} r_l\langle f_j|[\mP_k,\mP_l](1-\mP_j)\varkappa'(a_k-a_l) \nonumber \\
&\;\phantom{=} +\frac12 \sum_{k\neq j}^{\cN} (2+r_k)\langle f_j|\mP_k(1-\mP_j)\varkappa'(a_j-a_k) \nonumber \\
&= \frac12 \sum_{k\neq j}^{\cN}  (1+r_k)\langle f_j|\mP_k(1-\mP_j)\varkappa'(a_j-a_k)+\frac12\sum_{k\neq j}^{\cN}\sum_{l\neq j,k}^{\cN} r_l\langle f_j|[\mP_k,\mP_l](1-\mP_j)\varkappa'(a_k-a_l).
\end{align}
Hence, using that $r_j=1$, we get
\begin{align} 
&\langle C_1|+\langle C_2|+\langle C_{3,2}|+\langle C_{3,3}| \nonumber \\
&\simeq\frac12 \sum_{k\neq j}^{\cN}  (r_j+r_k)\langle f_j|\mP_k(1-\mP_j)\varkappa'(a_j-a_k)+\frac12\sum_{k\neq j}^{\cN}\sum_{l\neq j,k}^{\cN} r_l\langle f_j|[\mP_k,\mP_l](1-\mP_j)\varkappa'(a_k-a_l) \nonumber\\
&= \frac12\sum_{k=1}^{\cN}\sum_{l\neq k}^{\cN} r_l\langle f_j|[\mP_k,\mP_l](1-\mP_j)\varkappa'(a_k-a_l) \nonumber\\
&=\frac14\sum_{k=1}^{\cN}\sum_{l\neq k}^{\cN} (r_k+r_l)\langle f_j|[\mP_k,\mP_l](1-\mP_j)\varkappa'(a_k-a_l)
\end{align}
and from \eqref{eq:C31} it follows that
\begin{equation}	
\langle C_1|+\langle C_2|+\langle C_{3,1}|+\langle C_{3,2}|+\langle C_{3,3}|=\langle C_1|+\langle C_2|+\langle C_3| \simeq 0.
\end{equation}

We have thus shown that \eqref{eq:FjEjODEs} holds; a unique solution to the initial value problem consisting of \eqref{eq:FjEjODEs} with initial conditions $\langle F_j(0)|= 0$, $|E_{N+j}(0)\rangle=0$, $j= 1, \dots, N$ is given by $\langle F_j(t)|=0$, $|E_{N+j}(t)\rangle=0$, $j = 1, \dots, N$, for $t \in [0,\tau_{\mathrm{max}})$. It follows that $\mM$ and $\{a_j,|e_j\rangle,\langle f_j|\}_{j=1}^{\cN}$ uniquely solves the initial value problem of Lemma~\ref{lem:consistency} on $[0,\tau_{\mathrm{max}})$.

 \subsection{Proof of Lemma~\ref{lem:BT}}\label{app:lemBT}

In the system of equations given by \eqref{eq:sCMa}--\eqref{eq:sCMgh},
\eqref{eq:Mdotgen}, \eqref{eq:sCMconstraint1}, \eqref{eq:constraintPjQj}, and \eqref{eq:BT1}, the variables $\{a_j,|e_j\rangle,\langle f_j|\}_{j=1}^N$ and $\{b_j, \langle h_j|, |g_j\rangle\}_{j=1}^N$ can be swapped by Hermitian conjugation. Due to this symmetry, it suffices to verify the claim for the first set of variables, i.e., it is enough to show that \eqref{eq:sCMa} follows from \eqref{eq:sCMef}, \eqref{eq:sCMgh}, \eqref{eq:Mdotgen}, and \eqref{eq:BT1} subject to \eqref{eq:sCMconstraint1} and \eqref{eq:constraintPjQj}.

By differentiating the first set of equations in \eqref{eq:BTshort} with respect to time, we obtain 
\begin{equation} 
\label{eq:ddotajfj1}
\ddot a_j\langle f_j| =   \langle \dot f_j|(2\mB_j-\dot a_j) + 2\langle f_j|\dot \mB_j
\end{equation} 
where, here and below in this section, $j=1,\ldots,N$ (note that $r_j=+1$). Using \eqref{eq:Bj}, \eqref{eq:BTshort}, and \eqref{eq:sCM3b} and by nearly-identical calculations as in the proof of \cite[Proposition~2.1]{berntsonlangmannlenells2022}, \eqref{eq:ddotajfj1} may be written as \footnote{
Using \cite[Eq.~(2.31)]{berntsonlangmannlenells2022}, the right-hand side of of \eqref{eq:ddotaj} is seen to match that of \cite[Eq.~(2.29)]{berntsonlangmannlenells2022}. The new addition on the left-hand side, the term proportional to $\dot{\mM}$, comes from differentiating \eqref{eq:Bj} with respect to time and substituting the resulting expression into \eqref{eq:ddotajfj1}.
}
\begin{align}
\label{eq:ddotaj}
\langle f_j|(\ddot a_j-2\dot{\mM})  = &\;     \sum_{k\neq j}^\cN \sum_{l\neq j,k}^{\cN}(r_k+r_l)  \langle f_j| [\mP_k,\mP_l]\big(\alpha(a_j-a_k)-\alpha(a_j-a_l)\big)V(a_k-a_l) \nonumber\\
&\; +4\sum_{k\neq j}^\cN (1+r_k)\langle f_j| \mP_k\mP_j \alpha(a_j-a_k)V(a_j-a_k) \nonumber\\ 
&\; + 2\sum_{k\neq j}^\cN \sum_{l\neq j,k}^{\cN} \langle f_j|\big( r_kr_l \{\mP_k,\mP_l\} +r_l[\mP_k,\mP_l]) 
\big( \alpha(a_j-a_l) - \alpha(a_k-a_l) \big)V(a_j-a_k).   
\end{align}

We will rewrite \eqref{eq:ddotaj} using the identities 
\begin{equation}\label{eq:BId1}
\alpha(z)V(z)=-\frac12\big(V'(z)+\varkappa'(z)\big),
\end{equation}
which can be obtained by differentiating \eqref{eq:Id1} with respect to $z$, and \eqref{eq:alphaId2}. By inserting \eqref{eq:BId1} and \eqref{eq:alphaId2} into \eqref{eq:ddotaj} and rearranging, we obtain 
\begin{align}\label{eq:ddotaj2}
\ddot{a}_j \langle f_j| =&\; 2\langle f_j| \dot{\mM} -2\sum_{k\neq j}^{\cN} (1+r_k)\langle f_j |\mP_k\mP_j V'(a_j-a_k)  \nonumber \\
&\; +\sum_{k\neq j}^{\cN}\sum_{l\neq k}^{\cN}  (r_k-r_l) \langle f_j| [\mP_k,\mP_l] \big(\alpha(a_j-a_k)-\alpha(a_j-a_l)\big)V(a_k-a_l) \nonumber \\
&\; -2\sum_{k\neq j}^{\cN}\sum_{l\neq j,k}^{\cN} r_kr_l \langle f_j| \{\mP_k,\mP_l\}\big(\alpha(a_j-a_k)-\alpha(a_j-a_l)\big)V(a_k-a_l) \nonumber \\
&\; -2\sum_{k\neq j}^{\cN} (1+r_k)\langle f_j |\mP_k\mP_j \varkappa'(a_j-a_k) -\sum_{k\neq j}^{\cN}\sum_{l\neq j,k}^{\cN} \langle f_j| \big(r_kr_l \{\mP_k,\mP_l\}+r_l [\mP_k,\mP_l]\big)\varkappa'(a_j-a_k) \nonumber \\
&\; +\sum_{k\neq j}^{\cN}\sum_{l\neq j,k}^{\cN}r_kr_l \langle f_j | \{\mP_k,\mP_l\} \varkappa'(a_k-a_l) +\sum_{k\neq j}^{\cN}\sum_{l\neq j,k}^{\cN} r_l \langle f_j | [\mP_k,\mP_l] \varkappa'(a_k-a_l).
\end{align}
Since $V(z)$ is even, the double sums in the second and third lines of \eqref{eq:ddotaj2} each vanish by symmetry and since $\varkappa'(z)$ is odd, the first double sum in the final line of \eqref{eq:ddotaj2} vanishes by symmetry.

To simplify further, we recall the constraint \eqref{eq:sumconstraint4SH}, which inserted into the double sum in the fourth line of \eqref{eq:ddotaj2} gives 
\begin{multline}\label{eq:ddotaj3}
-\sum_{k\neq j}^{\cN}\sum_{l\neq j,k}^{\cN} \langle f_j| \big(r_kr_l \{\mP_k,\mP_l\}+r_l [\mP_k,\mP_l]\big)\varkappa'(a_j-a_k) \\
=\sum_{k\neq j}^{\cN}\langle f_j |\big(r_k \{\mP_k,\mP_j+r_k\mP_k\}+ [\mP_k,\mP_j+r_k\mP_k]\big)\varkappa'(a_j-a_k) \\
=\sum_{k\neq j}^{\cN}\langle f_j| \big(2\mP_k^2 +r_k \{\mP_j,\mP_k\}- [\mP_j,\mP_k]\big)\varkappa'(a_j-a_k) \\
=\sum_{k\neq j}^{\cN}(1+r_k)\langle f_j| \{\mP_j,\mP_k\}\varkappa'(a_j-a_k),
\end{multline}
where we have used $\mP_k^2=\mP_k$ and $\langle f_j| \mP_j =\langle f_j|$ in the final step (both are consequences of \eqref{eq:efSH}).
Hence, inserting \eqref{eq:ddotaj3} into \eqref{eq:ddotaj2} and simplifying, we have
\begin{align}\label{eq:ddotaj4}
\ddot{a}_j \langle f_j|=&\; 2\langle f_j|\dot{\mM} -2\sum_{k\neq j}^{\cN}(1+r_k)  \langle f_j|\mP_k\mP_j V'(a_j-a_k) \nonumber\\
&\; +\sum_{k\neq j}^{\cN}(1+r_k)\langle f_j| [\mP_j,\mP_k] \varkappa'(a_j-a_k) +\sum_{k\neq j}^{\cN}\sum_{l\neq j,k}^{\cN}r_l \langle f_j|  [\mP_k,\mP_l]\varkappa'(a_k-a_l). 
\end{align}
By symmetrizing the remaining double sum and using the fact that $\varkappa'(z)$ is an odd function, we see that
\begin{multline}
\sum_{k\neq j}^{\cN}\sum_{l\neq j,k}^{\cN}r_l \langle f_j|  [\mP_k,\mP_l]\varkappa'(a_k-a_l)= \frac12 \sum_{k\neq j}^{\cN}\sum_{l\neq j,k}^{\cN}(r_k+r_l) \langle f_j|  [\mP_k,\mP_l]\varkappa'(a_k-a_l) \\
=-\sum_{k\neq j}^{\cN} (1+r_k)\langle f_j|[\mP_j,\mP_k]\varkappa'(a_j-a_k)   +   \frac12 \sum_{k=1}^{\cN}\sum_{l\neq k}^{\cN}(r_k+r_l) \langle f_j|  [\mP_k,\mP_l]\varkappa'(a_k-a_l),
\end{multline}
where we have used $r_j=+1$. We arrive at
\begin{align}\label{eq:ddotaj5}
\ddot{a}_j \langle f_j|=&\;  -2\sum_{k\neq j}^{\cN}(1+r_k)  \langle f_j|\mP_k\mP_j V'(a_j-a_k) \nonumber \\
&\; +2\langle f_j|\dot{\mM}+\frac12\sum_{k=1}^{\cN}\sum_{l\neq k}^{\cN}(r_k+r_l) \langle f_j|  [\mP_k,\mP_l]\varkappa'(a_k-a_l).
\end{align}
The second line vanishes after inserting \eqref{eq:MdotSH}. Using that $(1+r_k)=2$ for $k=1,\ldots,N$ and $0$ otherwise, and multiplying \eqref{eq:ddotaj5} on the right by $|e_j\rangle$, we obtain \eqref{eq:sCMa}.

\bibliographystyle{unsrt}	
\bibliography{BLL6.bib}
\end{document}